\newif\ifanonymous
\newif\ifdraft

\anonymousfalse
\draftfalse
\InputIfFileExists{localflags}{}

\pdfoutput=1
 \documentclass[ 
     conference,
 ]{IEEEtran}

\ifdraft
\fi

\usepackage{versions}

\usepackage{xr}

\excludeversion{conf}
\includeversion{full}

\usepackage{amsthm}
\usepackage{amsmath}
\usepackage{amssymb}
\usepackage{stmaryrd}
\usepackage{mathabx} 
\usepackage[utf8]{inputenc}
\usepackage{xspace}
\usepackage{proof}
\usepackage{color}
\usepackage{paralist}
\usepackage[inline]{enumitem}
\usepackage{thm-restate}
\usepackage{booktabs} 
\usepackage{tikz}
\usepackage{pgfplots}
\usetikzlibrary{backgrounds,positioning,fit,calc}
\usepgflibrary{shapes.geometric}
\colorlet{rose}{red!20}
\colorlet{lima}{yellow!30}
\colorlet{light}{black!10}
\usetikzlibrary{fit}
\usetikzlibrary{tikzmark}
\usetikzlibrary{backgrounds,positioning,fit,calc}
\usepgflibrary{shapes.geometric}
\usetikzlibrary{decorations.pathreplacing,calc}
\usetikzlibrary{arrows,shapes,positioning,shadows,trees,automata,tikzmark}

\usepackage{msc}

\usepackage{listings} 
\usepackage{proof} 
\usepackage{subfig} 
\usepackage{algorithmic} 

\usepackage{inconsolata}

\usepackage{hyperref}
\usepackage{cleveref} 
\usepackage{verbatim}  
\usepackage{vector}  
\usepackage{algorithm}
\usepackage{algorithmic}
\usepackage{textcomp}    
\usepackage{amsfonts}
\usepackage{mathrsfs}
\usepackage{amssymb}
\usepackage{microtype}   
\usepackage{ps4pdf}
\usepackage{stmaryrd}
\usepackage{booktabs} 
\usepackage{listings}
\usepackage{proof} 
\usepackage{xspace}
\usepackage{multirow}
\usepackage{siunitx} 
\usepackage{tikz}
\usepackage{adjustbox}
\usepackage[utf8]{inputenc}
\usepackage{amsthm}
\usepackage{amsmath}
\usepackage{amssymb}
\usepackage{url}
\usepackage{setspace}

\usepackage{siunitx} 
\usepackage{listings}
\usepackage{enumitem}

\usepackage{algorithm}
\usepackage{algorithmic}
\usepackage{graphicx}
\usepackage{mathtools}
\usepackage{color}

\usepackage{booktabs}
\usepackage{multirow}

\usepackage{csvsimple}
\usepackage{cancel}
\usepackage{soul}

\usepackage[font={small}]{caption} 
\usepackage[
	backend=biber,
    sortlocale=en_GB,
    natbib=true,
    url=true,
    doi=false,
    eprint=false
]{biblatex}

\definecolor{gray}{RGB}{102,102,102}		
\definecolor{lightblue}{RGB}{0,102,153}		
\definecolor{lightgreen}{RGB}{102,153,0}	
\definecolor{bluegreen}{RGB}{51,153,126}	
\definecolor{magenta}{RGB}{217,74,122}		
\definecolor{orange}{RGB}{226,102,26}		
\definecolor{purple}{RGB}{125,71,147}		
\definecolor{green}{RGB}{113,138,98}		
 
\lstdefinelanguage{Proverif}{
  deletekeywords ={do},
  keywords = [1]{1,2,3,4,5,6,7,8,9,0},
  keywords = [2]{provider, dom, port, ip, as, bitstring, service, channel, regis, chan, com},
 keywords = [3]{smtp_client, smtp_server, resolver,root_server, dns_server, new, free, type, reduc, set, event, query, fun, forall, let, lookup, insert, in ,out, get, then, else, ;, ., |, process},
  keywordstyle = [1]\color{bluegreen},
  keywordstyle = [2]\color{lightgreen},
  keywordstyle = [3]\color{magenta},
  sensitive = true,
  commentstyle= \color{lightgray},
  morecomment = [l]{//},
  morecomment = [n][\color{orange}]{(*}{*)},
}

\newcommand{\p}{\bigg{\lVert}}
\pgfplotsset{compat=1.16} 

\makeatletter
\DeclareRobustCommand\bigop[1]{%
  \mathop{\vphantom{\sum}\mathpalette\bigop@{#1}}\slimits@
}
\newcommand{\bigop@}[2]{%
  \vcenter{%
    \sbox\z@{$#1\sum$}%
    \hbox{\resizebox{\ifx#1\displaystyle.9\fi\dimexpr\ht\z@+\dp\z@}{!}{$\m@th#2$}}%
  }%
}
\makeatother

\newcommand{\bigPar}{\DOTSB\bigop{\parallel}}

\AtEveryBibitem{
 \clearlist{address}
 \clearfield{eprint}
 \clearfield{isbn}
 \clearfield{issn}
 \clearlist{location}
 \ifentrytype{book}{
  \clearname{editor}
  }{}
\ifentrytype{inproceedings}{
  \clearname{editor}
  \clearfield{volume}
  \clearfield{url}
 \clearfield{month}
 \clearfield{date}
}{}
\ifentrytype{incollection}{
  \clearname{editor}
  \clearfield{volume}
  \clearfield{url}
 \clearfield{month}
 \clearfield{date}
}{}
\ifentrytype{proceedings}{
 \clearname{editor}
 \clearfield{volume}
 \clearfield{url}
 \clearfield{month}
 \clearfield{date}
}{}
\ifentrytype{article}{
 \clearfield{url}
 \clearfield{month}
 \clearfield{date}
}{}
}

\tikzset{
	basic/.style       = {draw, drop shadow, font=\small\ttfamily, rectangle},
	dom/.style         = {basic, rounded corners=2pt, thick, align=center, fill=white, draw=blue},
	ip/.style          = {basic, rounded corners=2pt, thick, align=center, fill=white, draw=green},
	org/.style         = {basic, rounded corners=2pt, thick, align=center, fill=white, draw=red},
	cc/.style          = {basic, rounded corners=2pt, thick, align=center, fill=white, draw=purple},
	as/.style          = {basic, rounded corners=2pt, thick, align=center, fill=white, draw=orange},
	label/.style       = {font=\small\sffamily, text centered},
	ex1/.style         = {basic, rounded corners=2pt, thick, align=center, fill=white, draw=black},
	prop/.style         = {basic, dashed, fill=white, draw=black}
}

\newcommand{\mathcmd}[1]{{\normalfont\ensuremath{#1}}\xspace}
\newcommand{\mathvar}[1]{\mathcmd{\mathsf{#1}}}

\newcommand{\mathname}[1]{\mathcmd{\text{\textrm{#1}}}}

\newcommand{\mathfun}[1]{\mathcmd{\mathit{#1}}}

\newcommand{\mathvalue}[1]{\mathcmd{\mathit{#1}}}

\newcommand{\mathlabel}[1]{\mathcmd{\textsf{#1}}}
\newcommand{\mathset}[1]{\mathname{#1}}

\newcommand{\textop}[1]{\relax\ifmmode\mathop{\text{#1}}\else\text{#1}\fi}

\newcommand{\PN}{\ensuremath{\mathit{PN}}\xspace}
\newcommand{\FN}{\ensuremath{\mathit{FN}}\xspace}

\newcommand{\Sign}{\ensuremath{\Sigma}\xspace}

\newcommand{\pout}{\ensuremath{\mathsf{out}}\xspace}
\newcommand{\pin}{\ensuremath{\mathsf{in}}\xspace}
\newcommand{\pif}{\ensuremath{\mathsf{if}}\xspace}
\newcommand{\pthen}{\ensuremath{\mathsf{then}}\xspace}
\newcommand{\pelse}{\ensuremath{\mathsf{else}}\xspace}
\newcommand{\plet}{\ensuremath{\mathsf{let}}\xspace}

\newcommand{\pevent}{\ensuremath{\mathsf{event}}\xspace}

\newcommand{\traces}{\ensuremath{\mathit{traces}}\xspace}

\newcommand{\theactualrule}[1]{\text{Please redefine the command
theactualrule.}}
\newcommand{\underscorethingy}[1]{\text{Please redefine the command
underscorethingy.}}

\newcommand{\senc}{\mathlabel{senc}}

\newcommand{\sdec}{\mathlabel{sdec}}

\makeatletter
\DeclareRobustCommand{\defeq}{\mathrel{\rlap{%
  \raisebox{0.3ex}{$\m@th\cdot$}}%
  \raisebox{-0.3ex}{$\m@th\cdot$}}%
  =}
\DeclareRobustCommand{\eqdef}{=\mathrel{\rlap{%
  \raisebox{0.3ex}{$\m@th\cdot$}}%
  \raisebox{-0.3ex}{$\m@th\cdot$}}%
  }
\makeatother

\newtheorem{definition}{\textbf{Definition}}

\newtheorem{lemma}{\textbf{Lemma}}

\newtheorem{theorem}{\textbf{Theorem}}
\newtheorem{example}{Example}

\newtheorem{soundassumption}{CS}
\newtheorem{completeassumption}{CC}

\theoremstyle{remark}
\newtheorem*{discussion}{\textit{Discussion}}

\newcommand{\ledot}{\mathrel{\ooalign{\hss\raise.200ex\hbox{$\cdot$}\hss\cr$\le$}}}
\newcommand{\gedot}{\mathrel{\ooalign{\hss\raise.200ex\hbox{$\cdot$}\hss\cr$\ge$}}}

\newcommand\abs[1]{\left\lvert#1\right\rvert}

\newcommand{\calA}{\ensuremath{\mathcal{A}}\xspace}

\newcommand{\calF}{\ensuremath{\mathcal{F}}\xspace}
\newcommand{\calG}{\ensuremath{\mathcal{G}}\xspace}

\newcommand{\calI}{\ensuremath{\mathcal{I}}\xspace}

\newcommand\calP{\ensuremath{\mathcal{P}}\xspace}

\newtheoremstyle{myrule} 
  {\topsep}
  {\topsep}
  {}
  {0pt}
  {\bfseries}
  {. }
  { }
  {\thmname{#1}\thmnumber{ #2} ---\thmnote{#3}}
\theoremstyle{myrule}
\newtheorem{myrule}{Rule}[section]

\newcommand{\set}[1]{\{#1\}}

\newcommand{\fullversionref}{symsound-full}
\usepackage{xr} 
\newcommand{\appendixorfull}[1]{%
\processifversion{conf}{the extended version~\cite[Appendix~\ref{F-#1}]{\fullversionref}}%
\processifversion{full}{Appendix~\ref{#1}}%
}

\newcommand{\countr}{\mathset{Cntry}}
\newcommand{\domain}{\mathset{Dom}}
\newcommand{\addr}{\mathset{IP}}
\newcommand{\AS}{\mathset{AS}}
\newcommand{\ASv}{\mathvalue{AS}}

\newcommand{\client}{\text{c}}
\newcommand{\server}{\text{s}}
\newcommand{\loc}{\mathlabel{LOC}}
\newcommand{\RTE}[1]{\mathlabel{RTE}(#1)}
\newcommand{\orig}{\mathlabel{ORIG}}
\newcommand{\A}{\mathlabel{A}}
\newcommand{\MX}{\mathlabel{MX}}
\newcommand{\NS}{\mathlabel{NS}}
\newcommand{\RES}{\mathlabel{RES}} 
\newcommand{\DNS}{\mathlabel{DNS}} 
\newcommand{\rns}{\mathlabel{RNS}}

\newcommand{\vpn}{\mathfun{\mathsf{nVPN}}}

\newcommand{\compr}{\mathfun{\mathsf C}}

\newcommand{\unconf}{\ensuremath{\mathsf{unconf}}}

\newcommand{\tls}{\ensuremath{\mathsf{nTLS}}}
\newcommand{\dnssec}{\ensuremath{\mathsf{nDNSSEC}}}
\newcommand{\dane}{\ensuremath{\mathsf{nDANE}}}
\newcommand{\strictValidation}{\ensuremath{\mathsf{nRFC7817}}}

\newcommand{\intrh}{\ensuremath{\mathsf{I^{DNS}}}}

\newcommand{\intrrd}{\ensuremath{\mathsf{I^R}}}
\newcommand{\intrd}{\ensuremath{\mathsf{I^{DNS}}}} 

\newcommand{\provider}{\ensuremath{\mathsf{Provider}}}
\newcommand{\prov}{\mathvar{Provider}}

\newcommand{\receiver}{\ensuremath{\mathsf{rcv}}}
\newcommand{\sender}{\ensuremath{\mathsf{snd}}}

\newcommand{\node}[1]{\mathvalue{#1}}
\newcommand{\ip}{\node{ip}}
\newcommand{\as}{\node{as}}

\newcommand{\Eps}{\mathcal{E}}
\newcommand{\Process}{\mathcal{P}}
\newcommand{\Frame}{\delta}

\newcommand{\frameabr}{\nu \Eps.\delta}
\newcommand{\substi}[2]{\left\{^{#1}/_{#2}\right\}}
\newcommand{\mulcup}{\cup^{\#}}

\newcommand{\vecv}{\overrightarrow{v}}

\newcommand{\vecin}[1]{\overrightarrow{#1}}
\newcommand{\Events}{\mathvalue{Events}}
\newcommand{\EventSig}{\Sign_\mathvalue{Event}}
\newcommand{\PT}{\mathcmd{\Pi}} 
\newcommand{\Ptraces}{\mathcal{T}^\PT}
\newcommand{\pt}{\mathvalue{pt}}
\newcommand{\traceset}{\mathcal{T}}
\newcommand{\setfun}[1]{\mathfun{set}(#1)}
\newcommand{\tracewo}[2]{#1\hspace{-0.26667em}\mid_{#2}}
\newcommand{\Inter}{\Sigma_{\cap}}

\newcommand{\pre}{\mathvalue{pre}}
\newcommand{\post}{\mathvalue{post}}
\newcommand{\postsec}{\mathfun{postseq}(\mathvalue{\pi})}
\newcommand{\perm}[1]{\mathfun{perm}(#1)}

\newcommand{\smtpserver}{\mathname{smtp-server}}
\newcommand{\smtpclient}{\mathname{smtp-client}}
\newcommand{\dnsres}{\mathname{res}}
\newcommand{\dnsns}{\mathname{ns}}
\newcommand{\dnsrns}{\mathname{rns}}

\newcommand{\reqchannel}{\mathname{req\_packet}}
\newcommand{\getreqpacket}{\mathname{get\_req\_packet}}

\newcommand{\anschannel}{\mathname{ans\_packet}}

\newcommand{\priv}{\mathsf{f_\mathname{priv}}}

\newcommand{\newstuff}[1]{\textcolor{orange}{\textbf{#1}}}
\newcommand{\rem}[1]{\newstuff{\st{#1}}}
\newcommand{\revised}[1]{\textcolor{black}{#1}}

\begin{document}

\ifanonymous \else
\author{%
\IEEEauthorblockN{%
Alexander Dax and Robert K\"{u}nnemann
}
\IEEEauthorblockA{CISPA Helmholtz Center for Information Security\\
Saarland Informatics Campus}
}
\fi
\title{%
    On the Soundness of Infrastructure Adversaries
}

\maketitle

\begin{abstract}
Companies and network operators perform risk assessment to
inform policy-making,
guide infrastructure investments
or to comply with security standards such as ISO 27001.
Due to the size and complexity of these networks, risk assessment
techniques such as attack graphs or trees describe the attacker with a finite set of rules.
This characterization of the attacker can easily miss attack vectors
or overstate them, potentially leading to incorrect risk estimation.

In this work, we propose the first methodology to justify a rule-based
attacker model.
Conceptually, we add another layer of abstraction on top of the
symbolic model of cryptography, which reasons about
protocols and abstracts cryptographic primitives.
This new layer reasons about Internet-scale networks and
abstracts protocols.

We show, in general, how the soundness and completeness of
a rule-based model can be ensured by verifying trace properties,
linking soundness to safety properties and completeness to liveness
properties.
We then demonstrate the approach
for a recently proposed threat model that quantifies the
confidentiality of email communication on the Internet, including DNS, DNSSEC, and SMTP. 
Using off-the-shelf protocol verification tools, we discover two 
flaws in their threat model. After fixing them, we show that it
provides symbolic soundness.
\end{abstract}

\section{Introduction}

The Internet is the primary medium for distributing entertainment,
news and knowledge and an important 
pillar to industrial commerce.
It is constructed from service providers interoperating according to
several protocols. Many of them were conceived before the Internet
was even considered a Mass Medium~\cite{10.1111/j.1083-6101.1996.tb00174.x};
they were hence designed to be fast and service-oriented, whereas
security was a second thought. Trust between service providers at
different protocol layers is thus an implicit assumption, making it
difficult to estimate potential attacks' impact. 

High-profile attacks, e.g. on routing~\cite{Marczak:GreatCannon} or
name resolution~\cite{myetherwalletMESSAGEOURCOMMUNITY2018} are a painful reminder of these trust assumptions.
They also highlight the slow adoption
of security protocols, which were developed only post-hoc,
to mitigate some of these issues.
Even for TLS~\cite{rescorla2018transport}, which enjoys high
popularity, adoption was and remains slow.
According to Qualys Labs~\cite{Pulse}, 6\% of all websites still support
SSL 3.0, which is exploitable in various manners and was deprecated in
2015.
Moreover, security protocols rely on trust assumptions and
a complicated interplay between routing, name resolution, and the
application layer.
An example is RFC 7817~\cite{rfc7817}, which defines certificate validation for
email transport. It mandates that the certificate contains the email
domain (the part after the '@') and not just the target server's domain name, as a name resolution attacker can easily manipulate the latter.
Large-scale attacks thus rarely
exploit previously unknown flaws in a single protocol, but instead
target their deployment in the wild. 

Despite the effort put into securing individual protocols and
cryptographic primitives in the past decades,
worldwide attacks like the
Great Cannon~\cite{marczak2015china} or spying systems like
PRISM~\cite{greenwald2013nsa} exploit weak components and (the
absence of) trust anchors in the infrastructure.
To analyze an infrastructure like the Internet, with broken legacy
protocols, unstable trust assumptions, and varying degrees of
centralization on different layers,
a high-level approach is necessary.

\textit{Risk assessment} originates in the formal assessment of
potential failures in large infrastructures like power plants.
Techniques like fault trees provide a systematic method for
identifying and minimizing potential risks. They were soon adopted for
IT infrastructure.
These techniques usually consider the severity of known vulnerabilities and
some valuation of critical assets.
The problem size grows with the size of the network. 
Therefore, most of these techniques formalize the threat model as
a set of rules.
Those techniques include planning, attack graphs (which were derived from fault graphs), and game-based models.\footnote{For other techniques, consider the study by \citet{wang2007study}.}

While these analyses are formal and well justified, the rules
themselves are not formally justified. It is not safe to assume that
the set of rules is comprehensive. Thus the analysis may
miss potential attack vectors.
There is a surprising similarity to the soundness of the Dolev-Yao
model.
Abadi and Rogaway's seminal paper on computational
soundness~\cite{abadi_reconciling_2002}
considered the soundness of a such a rule-based symbolic attacker on
protocols in the computational model. 
Likewise,
our focus is on the soundness of
a rule-based attacker, the infrastructure attacker (IA), but in the
symbolic model instead of the computational model.
In both
cases, the need for further abstraction is driven by the complexity of
the problem (infrastructure analysis/protocol analysis) but
requires justification.



\subsection*{Contributions}

\begin{enumerate}
    \item We define proof obligations for 
        the correctness of an infrastructure attacker in the STRIPS
        framework for planning as a set of trace properties.
        We show that soundness can be proven by verifying  
        \emph{safety properties}, and correctness by
        verifying \emph{liveness properties}.
    \item We apply this definition to an IA model for email
        communication~\cite{speicher2018formally} and establish its
        soundness (barring some minor flaws).
    \item We show how to automate the proof of this trace properties
        by over-approximating all possible instantiations of the
        IA model with a single process.
        The protocol transformations we introduce to this end are of
        independent interest, as they can help to reduce drastically the size of processes that model an adversary with limited access to network traffic. 
    \item We show various authentication properties of SMTP in
        conjunction with DNS, DNSSEC, and a simple resolver model in
        ProVerif.
        As a by-product, this model provides the first automated
        verification result for authentication in DNSSEC.

\end{enumerate}


\section{Related Work}
\label{sect:relwork}

\subsubsection{Risk assessment techniques}

The most popular techniques for the analysis of IT infrastructure are
attack graphs \cite{zeng2019survey} and trees\cite{schneier_1999}, see
\cite{mantel2019meaning} for a recent survey.
They originate in 
risk assessment and reliability analysis for critical infrastructures.
Fault tree analysis~\cite{clifton1999fault}
was used,
e.g. for analyzing nuclear power plants or
military missile control systems.
Attack graphs and trees have been used to assess risks in
forensic examination~\cite{liu2012using},
network security~\cite{kotenko2006attack,phillips1998graph} 
or cloud infrastructures~\cite{Nawaf}.
Used naïvely, both techniques suffer from the state explosion
property. Luckily, a large body of work is devoted to improving
performance, e.g. 
generating of minimal attack graphs~\cite{ghosh2009intelligent},
distributing attack graph generation~\cite{kaynar2015distributed}
or the efficient representation of network defenses~\cite{ingols2006practical}.

More recently, planning was considered as an alternative
technique with great benefits in terms of
performance~\cite{ghoshPlannerbasedApproachGenerate2012}.
Planning is one of the oldest sub-areas of AI
and benefits from being a well-studied research field with a large
community and a focus on optimizing performance.
Compared to various semantics for attack graphs, there is a fairly wide
agreement on the STRIPS framework~\cite{strips}.
Planning was used for attack graph generation \cite{hendler1990ai},
network analysis \cite{boddy2005course},
penetration testing~\cite{obes2013attack},
and internet
infrastructure analysis \cite{speicher2018stackelberg}. Additionally,
the popular attack graph formalism can be translated into a planning
problem \cite{hoffmann2015simulated}.

\subsubsection{Infrastructure analysis}

Until recently, these approaches were used to analyze local networks
or the public infrastructure unrelated to information security.
Presumably, this was due to the problem size associated with
large-scale infrastructures like the Internet.
\citet{frey2016bends} conducted one of the first Internet-scale
infrastructure assessments in terms of security evaluation. They
investigated the Border Gateway Protocol deployment looking into
potential threats and vulnerabilities.
\citet{cispa1091} present a technique that models services, providers,
and dependencies on the Internet as a property graph, establishing
a high-level IA model. This model is used to reason about
dependencies between services and infrastructure providers
and how these dependencies can be exploited to impact a large amount
of end users.
They
conduct a large-scale case study by using a simple tainting-style
propagation technique in a graph database highly optimized for
reachability queries. 
They studied several
attack scenarios like email-sniffing and DDOS caused by the
distribution of malicious JavaScript.
More recently, \citet{speicher2018formally} introduced the first
deployment analysis on a global level, evaluating
various measures to secure the email infrastructure against
large-scale attacks.
They employ Stackelberg planning~\cite{speicher2018stackelberg}, which
is a two-stage planning technique that computes all
defender plans that are Pareto-optimal with respect to their cost
and the worst-case impact of an attacker.

To our knowledge, all these assessment approaches have only informally
justified their threat models. Given the high abstraction level of
their reasoning, validation using formal analysis techniques is
necessary.
\citet{sheyner_automated_2002} propose the use of symbolic model
checking to generate attack graphs from a finite state machine that
represents the network. Here, state transitions correspond to atomic
attacker steps which themselves require justification. 
This approach is neither applicable to larger networks (because of the
aforementioned state-explosion property) nor does it provide the
desired level of justification (as network attackers are too complex for
finite state machines).

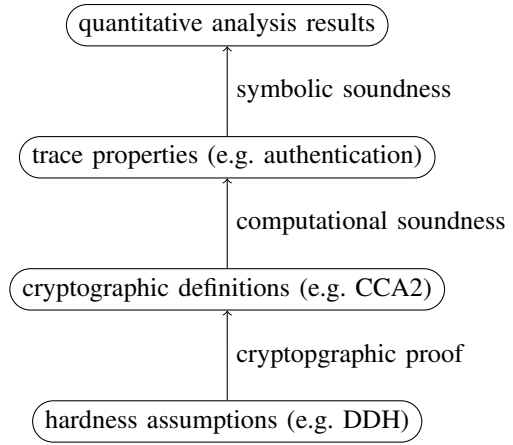
\begin{figure}
\centering
\begin{tikzpicture}[
    node distance = 5em
  , mbox/.style = { draw,rounded rectangle}
    ]

    \node (ass) [mbox]  {hardness assumptions (e.g. DDH)};
    \node (cr) [mbox, above of=ass]  {cryptographic definitions (e.g. CCA2)};
    \node (tr) [mbox, above of=cr]  {trace properties (e.g.
        authentication)};
    \node (q) [mbox, above of=tr]  {quantitative analysis results};

    \path[->]
        (ass) edge node[right] {cryptopgraphic proof} (cr)
        (cr) edge node[right] {computational soundness} (tr)
        (tr) edge node[right] {symbolic soundness} (q)
        ;
    
\end{tikzpicture}

\caption{Relation between different levels of abstraction}
\label{fig:relation-between-different-levels-of-abstraction}
\end{figure} 

\subsubsection{The analogy to computational soundness}
\label{sec:analogy}

In contrast to cryptographic primitives like encryption or signatures,
larger cryptographic protocols are typically analyzed in the Dolev-Yao
model~\cite{dolev1983security},
where cryptographic primitives (short: crypto primitives)
are abstracted with a term
algebra.
Proofs in the computational model are possible, but become
prohibitively complex due to the need to reason about probabilistic
behavior and runtime. Mechanization becomes incredible
difficult~\cite{Barthe2011Computer-Aided-} and manual proofs can easily
miss details.
By contrast, the abstraction of cryptography by a term algebra has
enabled the development of fully automatic and semi-automatic protocol
verifiers~\cite{blanchet_efficient_2001,hutchison_tamarin_2013}
that can handle highly complex protocols, e.g.
TLS~1.3~\cite{bhargavan2017verified,cremers_comprehensive_2017}.

To bridge both worlds, \citet{abadi_reconciling_2002} introduced
\textit{computational soundness}, justifying the use of the symbolic
model for protocol analysis. Gaining the advantages of automation in
the symbolic model, and the stronger guarantees of the computational
model, the notion of computational soundness is seen as a massive
milestone in protocol verification.  With this work, we want to extend
on this stack in an analogous fashion, and introduce the notion of
\textit{symbolic soundness} relating the infrastructure adversary
to the symbolic model in a similar fashion. This is depicted in
Fig.~\ref{fig:relation-between-different-levels-of-abstraction}.

\begin{table}
\centering
\small
\begin{tabular}{p{20mm}p{29mm}p{26mm}}
\toprule
& computational sound.
& symbolic soundness
\\
\midrule
threat model 
& network attacker 
& infrastructure att. 
\\
assumption 
& perfect cryptography 
& perfect protocol
\\
high-level semantics
& term algebra + process calculus
& planning (STRIPS)
\\
low-level semantics
& probabilistic Turing Machines
& term algebra + process calculus
\\\midrule
\emph{Proof strategy}
\\
Fix a set of
& crypto primitives
& protocols
\\
conforming to 
& crypto definitions (e.g. CCA2 for Cramer-Shoup).
& trace properties (e.g. authenticity for TLS).
\\
For all 
& protocols 
& network topologies
\\ 
map each
& computational traces
& protocol trace
\\
from
& comp. executions (interpreted by TM)
& processes (compiled from the topology)
\\
to a
& (symb.) protocol trace.
& plan.
\\
\bottomrule
\end{tabular}
\caption{Comparing computational soundness and symbolic soundness.}
\label{tab:symcompsound}
\end{table}

To attain this goal, we lift this approach from the protocol level to
the infrastructure level. To readers familiar with computational
soundness, Tab.~\ref{tab:symcompsound} can help to support this
analogy.

We formalize the IA model as a planning problem in
the STRIPS formalism~\cite{strips}. The rules represent an infrastructure
attacker that can selectively corrupt parts of the infrastructure, but
assumes that protocols themselves are secure. 
Dolev-Yao models, by contrast, represent a network attacker, but
assume the crypto primitives to be perfect.
To reason about the validity of the model w.r.t.\ those assumptions,
they need to be formalized. As they are implicit in the respective
semantics (STRIPS / process calculus with term algebra), a low-level
semantics is necessary to state these assumptions. Assumptions about
protocols are stated in the Dolev-Yao model, usually within
a process calculus with a term algebra. Assumptions about
cryptographic primitives are stated as asymptotic probability bounds on the
probability of a runtime-bounded Turing machines winning some game.

Symbolic soundness
asserts that, when compiling a given infrastructure model into
a process, all symbolic traces that this process allows can be mapped to
attacker plans in the planning model. This is \revised{structurally} similar to computational
soundness, which ensures that no computational attacks are missed by
mapping all computational executions to symbolic traces (with
a negligible failure probability). We can thus show that no
\revised{symbolic} attack is
missed by the IA model, provided the protocols are working as
intended.
\revised{To be clear: symbolic soundness does not imply computational
    soundness, but both combine. If a symbolic soundness
    result asserts the absence of attacks w.r.t.\ to some symbolic
    model, then a computational soundness can extend this result to the
computational model of cryptography. This requires that the
computational soundness result supports the
cryptographic primitives and process calculus used by the symbolic
soundness result.}

Much like results in computational soundness, symbolic soundness
results apply to a fixed set of primitives, in their case: protocols.
We model an infrastructure 
consisting of 
DNS, DNSSEC and SMTP, 
using a dialect of
the applied-$\pi$ calculus~\cite{abadi2001mobile}.

\subsubsection{Symbolic completeness and liveness properties}

In contrast to computational soundness,
the completeness of an IA model is equally important to the analysis.
A quantitative analysis, e.g. counting the number of affected hosts,
is incorrect if the IA model overestimates the protocol attacker's
capabilities.
In the case of Stackelberg planning, this might lead to the
proposition of suboptimal countermeasures and, if the defender budget is
fixed, to an allocation that is not optimal for security.
We therefore define \textit{symbolic completeness} and show a set of
conditions that implies the completeness of this model.
Unfortunately, one of these is a liveness property,
i.e., a property of
the form: `the [protocol] eventually enters a desirable
state'~\cite{lamport1977proving}.
Practically all 
protocol verification tools~\cite{blanchet_efficient_2001,hutchison_tamarin_2013,cremers2008scyther}
in the unbounded model
cover only safety
properties, i.e., properties of the form `the protocol never enters
a bad state'.
Hence, there is currently no support for the verification of liveness
properties such as ours.
We further elaborate on this topic in
\cref{sec:defsymcomp}.

\subsubsection{Analysis of DNSSEC}

To our knowledge, our case study provides not only the first automated result w.r.t.\ an infrastructure attacker, but
also the first automated verification result for DNSSEC\@.
\citet{chetiouiFormalVerificationConfidentiality2019}
investigate (weak) secrecy in E-DNSSEC, a variant of DNSSEC that
adds encryption, in ProVerif.
\citet{kammullerVerificationDNSsecDelegation2014}
also cover authentication in a handwritten, but automatically verified
proof in Isabelle/HOL\@.

\section{Background: Automated Planning}
\label{sec:back}

A planning task is usually described in the STRIPS framework~\cite{strips}. Here,
$\PT = (\calP, \calI, \calA, \calG)$ 
is defined over a high-level
representation of the world in which each state $\sigma$ is built over
a set of \textit{state propositions} \calP.
$\calI \subseteq \calP$ is the initial state
and the task is to
reach a \textit{goal states} in $\calG \subseteq 2^{\mathcal{P}}$.
A set of \textit{ actions} \calA over \calP defines transitions between states. 
Actions are described as a triple (\textit{pre, del, add}) where
\textit{pre} $\subseteq$ \calP  is the set of preconditions needed in
the current state to make the action  applicable, \textit{del}
$\subseteq$ \calP tells which proposition will be  deleted in the
transition to the next state whereas \textit{add} $\subseteq$  \calP
tells which propositions are added. In  \textit{classical planning},
we assume that all actions have a  deterministic effect and that the
initial state of the world is known from
beginning.
A state $\sigma$ is \textit{reachable} from \calI, if there is a sequence of actions $a_0 \cdots a_k$, which can be applied to \calI one after another resulting in $\sigma$. We call this sequence of actions a plan $\pi$ to reach $\sigma$.
The basic idea behind planning is to find a sequence of actions, s.t. their application starting from the initial state \calI leads to one of the goal states in \calG.
\revised{%
\citeauthor{speicher2018formally}, e.g., consider the initial state as
the nodes that an attacker controls from the start, e.g., different
nation-state adversaries or companies abusing power. Goal states are
valuable assets that need to be protected, e.g., the largest mail
providers within some country.
}
Over the years, several variations of
automated planning have been developed,
with different 
 modeling assumption and resulting complexity classes for plan existence, worst-case runtime, etc.

We focus on classical planning for the ease of presentation.
Our approach easily transfers to probabilistic planning when considering uncertainty about the initial set-up or effect probabilities
as model parameters. We cannot justify these parameters via protocol
verification (which is typically possibilistic) or cryptographic
reasoning in general. These parameters model uncertainty
about the attacker's capabilities and intentions. They are thus
outside the current scope of formal analysis in security.
Our infrastructure attacker is described by actions that have only
positive preconditions and postconditions\revised{, i.e., they are described as
    pairs $(\pre,\post)\in \calP^2$ instead of tuples
$(\pre,\mathit{del},\mathit{add})$.}
Such planning tasks are
called \textit{delete-relaxed} or \textit{monotonous} and are easier
to solve. Delete-relaxed planning aligns with the implicit assumption
that attackers only gain assets in attack graph
analysis~\cite{ammann2002scalable}.

Stackelberg planning \cite{speicher2018stackelberg} elevates this
form of analysis to a two-player planning task in an
attacker/defender scenario. In this scenario, the defender tries to
implement mitigation strategies to limit the impact of the worst-case
attacker strategy.
A Stackelberg planning task differs from a classical task by dividing
the set of actions into \textit{leader} (or attacker) actions
$\mathcal{A}^{\mathcal{L}}$ and into \textit{follower} (or defender)
actions $\mathcal{A}^{\mathcal{F}}$. Further, the goal states are now
defined for the defender, namely defender/follower goals
$\calG^\mathcal{F}$.
In this setting, an attack is composed of attacker actions,
but applies to a world state where the defender has applied a plan
composed of defender actions to the initial state. 
Every attack is annotated with some attacker reward, which
depends on the severity of the attack (e.g. number of corrupted
connections due to the attack). Defender actions come with a cost.
The Stackelberg planning algorithm computes the set of Pareto-optimal
pairs of attacker and defender plans.
For the soundness of the attacker model, it is enough to consider the
classical planning task where the follower actions are removed and only
the attacker goal is considered, but the initial state can be any
state reachable via defender actions. 
In our analysis, the initial state is, in fact, arbitrary.

\section{Symbolic Soundness and Completeness}
\label{motsymsound}
 
We introduce the concepts of \emph{symbolic soundness} and
\emph{symbolic completeness}, which relate the infrastructure
adversary model (formalized as a planning problem) to the Dolev-Yao
model~\cite{dolev1983security}. 
Our approach applies to security properties that can be expressed as
trace properties.
We start by introducing the necessary notation and concepts. Then we
introduce the conditions under which
symbolic soundness and
symbolic completeness hold. Finally, we prove these statements.

\subsection{Notation}
\label{sec:nota}


For a sequence $s\in \Sigma^*$, let $\setfun{s}$ be the set of elements in $s$.
For $e\in\Sigma$, $s \circ e$ denotes the concatenation with $e$.
For $S\subseteq \Sigma$, 
$\tracewo{s}{S}$ is $s$ with every element outside $S$ removed.

The IA model is formalized in terms of a finite set of planning
actions. 
We define
$\postsec = \post_0, \post_1, \ldots, \post_n$ 
to be the sequence of postconditions of some
plan $\pi = (\pre_0,\post_0), \ldots, (\pre_n,\post_n)$.
We define a planning trace of some plan $\pi$ as a sequence 
$\pt = s_1, s_2, \ldots, s_n$, where for all $i \in \{1, \ldots, n\}$,
$s_i \in \perm{\post_i}$ is some permutation of $\post_i$.
If all postconditions in $\pi$ are singleton sets,
it has only one $\pt$.
Let $\Ptraces(\sigma)$ be the union of all planning traces reaching
$\sigma$, and (with slight abuse of notation)
$\Ptraces \subseteq \calP^*$ denote the union of planning traces
over all states.

For generality, symbolic soundness and completeness are formulated
independent of the process calculus. 
We assume a set of traces
of form $\traces = \Events^*$ 
that represents the possible behavior of
a protocol and is usually specified by encoding it into a process.
To simplify the presentation and avoid introducing a mapping function,
we assume a non-empty intersection between
predicates $\calP$ and events $\Events$. 
Our aim is to match planning traces and protocol traces on this
intersection, which we denote by $\Inter$. 
Typically, the
predicates/events in this set signify the corruption of some party or
the partial compromise of certain infrastructure services
(cf.\ Table~\ref{table:corrupt} for examples).
We hence call them corruption predicates.
%
\begin{definition}{$\approx$-equivalence}\\
\label{def:wteq}
Let $\approx = (\Ptraces \cup \traces)^2$ 
s.t.\
$s \approx t \iff \tracewo{s}{\Sigma_{\cap}}$ = $\tracewo{t}{\Sigma_{\cap}}$.
\end{definition}
When all predicates $\calP$ are contained in $\Inter$,
our approach can be seen as a refinement, where planning traces
provide an abstract view on protocol traces.

\subsection{Symbolic Soundness}
\label{sec:defsymsound}
\label{sec:asusymsound}

We define the symbolic soundness of a planning task w.r.t.\ a set of
traces. We will then provide sufficient conditions for
this property. Two of them can be checked statically on the planning
problem; the third holds for most process calculi. The fourth induces
a set of trace properties that can be discharged to protocol
verifiers.
We say that a planning task is sound if any behavior of the
protocol, e.g. an attack, is represented in the planning task.

\begin{definition}[Symbolic Soundness]\label{def:symsound}
A planning task $\PT$ is symbolically sound 
w.r.t.\ a set of traces $\traceset\subset \Events^*$, 
if for every trace $t\in\traceset$, 
there is a planning trace $\pt \in \Ptraces$
s.t.
$\pt \approx t$. 
\end{definition}

\revised{Symbolic soundness provides guarantees with respect to the
    Dolev-Yao model. In case the Dolev-Yao model (represented by
    $\traceset$) is covered by a computational soundness result, these
    guarantees may translate to the computational model, but a~priori,
these are guarantees in a symbolic model of cryptography.}
We now state and discuss sufficient conditions for soundness for an
arbitrary but fixed planning problem 
$\PT = (\calP, \calI, \calA, \calG)$ 
and a set of traces
$\traceset$.

%
\begin{soundassumption}
\label{asu:postsing}
All postconditions are singleton.
\begin{align*}
\forall a = (\pre, \post) \in \mathcal{A}: |\post| = 1
\end{align*}
\end{soundassumption}
%
%
\begin{discussion}
This condition is true w.l.o.g. for all monotonic planning
tasks~\cite{bylander1994computational} \revised{whose postconditions are
positive.}
Any action with $n > 1$ postconditions $a = (\pre, \{c_1, ..., c_n\})$ 
can be split 
into $n$ actions 
$a_i = (\pre, \{c_i\})$ without losing completeness or soundness.
Since we never delete any information from the state, each plan where
$a$ occurs can be recovered by substituting $a$ with the sequence
$a_1,\ldots,a_n$.
Conversely, we can apply $a$ whenever any plan contains some $a_i$.
\end{discussion}

%
%
\begin{soundassumption}
\label{asu:postcomp}
All corruption predicates are reproducible in the planning model.
\begin{align*}
\forall e \in \Inter. \exists (\pre,\post) \in \mathcal{A}\ :\ \post = \{e\}
\end{align*}
\end{soundassumption}
\begin{discussion}
    This condition is largely technical.
    Note first that $\post$ is singleton by
    \autoref{asu:postsing}.
    The set of corruption predicates $\Inter$ should be chosen to represent all events where planning traces and protocol traces ought to match. Hence the planning model must be able to produce them.
    Furthermore, 
    any planning task can be transformed 
    so that all predicates in $\calP$ appear in
    some action's postcondition: we let $\calI=\emptyset$ and add an action that reaches the initial state. Now all actions with
    preconditions that do not appear in any postcondition can be
    removed and $\calP$ be set to the union of postconditions.
    As $\Inter \subseteq \calP$, this implies \autoref{asu:postcomp}.
\end{discussion}

\begin{soundassumption}
\label{asu:prefixclosed}
The set of traces is prefix-closed.  For any $k>0$
\begin{align*}
    \forall e_1,\ldots,e_k. (e_1,\ldots,e_k)\in\traceset \implies
    (e_1,\ldots,e_{k-1})\in \traceset
\end{align*}
\end{soundassumption}
\begin{discussion}
    This condition concerns the semantics of the process calculus.
    It holds for 
    ProVerif~\cite{blanchet_efficient_2001},
    Tamarin~\cite{hutchison_tamarin_2013} and
    Scyther~\cite{cremers2008scyther}.
\end{discussion}

\revised{
\begin{soundassumption}
\label{asu:presubset}
The production of predicates in $\Inter$ is \emph{not} dependent on predicates
outside of this set.
\begin{align*}
    \forall \post \in \Inter. \forall (\pre, \{\post\}) \in \mathcal{A}. \forall f \in \pre: f \in \Inter
\end{align*}
\end{soundassumption}
%
%
\begin{discussion}
	The corruption predicates $\Inter$ are used to describe the security
	model in both languages. With this condition we restrict the model to 
	independent of predicates outside of $\Inter$.
	We refrain from forbidding predicates outside of $\Inter$ in the planning
	model as they appear to be useful in quantitative tasks. For instance,  
	counting occurrences of specific corruption predicates can be
	essential in a quantitative analysis. Such a model would be depended on
	predicates in $\Inter$ but not vice versa.
\end{discussion}
}

%
%
\begin{soundassumption}
\label{asu:postsplit}
Let 
$\mathcal{A} = \mathcal{A}_{c_1} \uplus \mathcal{A}_{c_2} \uplus \cdots \uplus \mathcal{A}_{c_n}$ 
be the set of actions, 
partitioned into disjunct sets $\mathcal{A}_{c_i}$,
where there is exactly one set per postcondition $\set{c_i}$. (By $\autoref{asu:postsing}$, all postconditions are
singleton.)
We assume that, whenever  a postcondition $c_i$ appears in a trace,
then a matching precondition appears, too, namely the precondition of
\emph{some} action in $\calA_{c_i}$.
\begin{multline*}
    \forall i \in \{1..n\},  t \in \traceset:
c_i \in t \wedge c_i \in \Inter \implies \\
\exists a = (\pre_i, \set{c_i}) \in \mathcal{A}_{c_i}:
 \forall g \in \revised{\pre_i}: g \in t.
\end{multline*}
\end{soundassumption}
%
%
\begin{discussion}
    This property is a safety property 
    and
    can be shown 
    using any protocol verifier that handles correspondence
    properties, e.g. 
    Tamarin~\cite{hutchison_tamarin_2013} or Scyther~\cite{cremers2008scyther}.
    In Section~\ref{sec:caseres}, we use ProVerif
    \cite{blanchet_efficient_2001} to this end.
\end{discussion}

The following theorem establishes the soundness of this approach:
\begin{theorem}\label{the:symsound}
    If
    \autoref{asu:postsing},
    \autoref{asu:postcomp},
    \autoref{asu:prefixclosed},
    \autoref{asu:presubset}
    and \autoref{asu:postsplit} hold, then
    $\PT$ is symbolically sound.
\end{theorem}
%
%
%
%
\begin{proof}
Proof by induction over the length of $\tracewo{t}{\Inter}$.

    \noindent
\emph{Base case $\abs{\tracewo{t}{\Inter}} = 0$:}
Let $\sigma$ = \calI. Then $\Ptraces(\sigma)$ = $\setfun{()}$. For the empty 
trace $t$, it holds that $\tracewo{t}{\Inter}$ = $()$ $\in$ $\Ptraces(\sigma)$.

\noindent
\emph{Inductive step:} Let $\abs{\tracewo{t}{\Inter}} = k+1$.\revised{
Let $\tracewo{t}{\Inter} = (e_1 e_2 .. e_{k+1})$. By \autoref{asu:prefixclosed} and the inductive hypothesis,
there is a 
$\tracewo{t_k}{\Inter} = (e_1 e_2 .. e_k)$ and a planning trace
$pt_k$, with $pt_k \approx t_k$.
From $pt_k$ 
we can infer that there exists a reachable state (of $\PT$) $\sigma_k$ 
with $\tracewo{\sigma_k}{\Inter} = \{e_1, e_2, .. e_k\}$.}

\revised{
By  \autoref{asu:postsing}, we get that all postconditions of any action in $\calA$
are singleton sets.}
By \autoref{asu:postcomp}, 
there exists an action $a \in \mathcal{A}$ with $a = (\pre_a, \post_a)$ and 
$ \post_a = e_{k+1}$.
Let $\mathcal{A}_{e_{k+1}}$ be the partition of all of $A$ containing all actions with postcondition $e_{k+1}$.
\revised{As $e_{k+1} \in \Inter$, by \autoref{asu:presubset} all preconditions are in $\Inter$, too.}
\revised{
By \autoref{asu:postsplit}, there exists an action
$a^* = (\pre, \{e_{k+1}\})$ s.t.\ for all $g \in \pre: g \in t_k$.
As $\pre\subseteq \Inter$, all $g\in\pre$ are in
$\setfun{\tracewo{t_k}{\Inter}}= \setfun{\tracewo{pt_k}{\Inter}}$
and thus in $\sigma_k$.
Applying $a^*$ to the state, we get $\sigma_{k+1}$ with 
$\tracewo{\sigma_{k+1}}{\Inter} = \{e_1, e_2, .. e_k,e_{k+1}\}$.}

Finally, we can conclude that there exists a planning trace $\pt_{k+1} \in \Ptraces(\sigma_{k+1})$ s.t. $t \approx \pt_{k+1}$, namely
$t \approx e_1 ... e_k e_{k+1} \approx \pt_k \circ e_{k+1} = \pt_{k+1}$.
\end{proof}

\subsection{Symbolic Completeness}
\label{sec:asusymcomp}
\label{sec:defsymcomp}

The complementing property to symbolic soundness is symbolic
completeness.
It ensures that the planning model does not introduce spurious attacks
that cannot occur in the protocol model.
Planning problems are frequently used to perform a quantitative
assessment of, e.g. the number of reachable goal states or the
probability of reaching certain assets. The correctness of such an
assessment relies on symbolic soundness \emph{and} symbolic completeness.
This is in contrast to computational completeness,
which is of little interest as long as the symbolic model is good
enough to provide verification results.

%
\begin{definition}[Symbolic Completeness]
\label{def:symcomp}
A planning task $\PT$ is symbolically complete 
w.r.t.\ $\traceset$
if 
for every  planning trace $\pt$,
there is a trace $t\in\traceset$ s.t.\ $\pt \approx t$. 
\end{definition}

We provide an additional assumption that ensures symbolic
completeness.
Unfortunately, it is a \emph{liveness property}, i.e., a property of
the form: `the [protocol] eventually enters a desirable
state'~\cite{lamport1977proving} and cannot be verified by the current
generation of protocol verifiers.

%
%
\begin{completeassumption}
\label{asu:hyper}
If an action is available and the trace contains the necessary
preconditions, then the trace can be extended so it contains this
action's postcondition.
\begin{align*}
    \forall t \in \traceset,& a = (\{p_1, ..., p_n\}, c) \in \mathcal{A}:
    \\
& \revised{c \in \Inter \wedge (\{p_1, ..., p_n\})} \subset \setfun{t} \implies \\
& \exists t' \in\traceset: t' = t \circ t_r \wedge (\setfun{t_r} \cap \Inter) = \set{c}.
\end{align*}
\end{completeassumption}
\begin{discussion}
%
Lamport~\cite{lamport1977proving} informally describes such properties 
as liveness properties. Note that here, the `desirable state' is an
additional attack step.
As we only consider finite traces, 
Alpern and Schneider's definition of liveness~\cite{alpern1985defining},
---
which is well known because it 
decomposes trace properties into safety and liveness properties
---
does not classify \autoref{asu:hyper} as a liveness
property.\footnote{%
According to their definition,
`no partial execution is irremediable since if some partial execution were irremediable, then it would be a ``bad thing''.'
}
Other characterisations do, see \textcite{kindler1994safety} for a survey.

Nevertheless,
state-of-the-art protocol verifiers in the unbounded
setting~\cite{blanchet_efficient_2001,hutchison_tamarin_2013,cremers2008scyther}
only support the specification of properties of the form 
$\forall t\in\traceset. \varphi(t)$ where $\varphi$ is a property that
is protocol-agnostic, i.e. invariant w.r.t.\ $\traceset$.
This prohibits a direct encoding of \autoref{asu:hyper}.

Backes, Dreier, Kremer, and Künnemann propose an encoding of liveness
properties for Tamarin that allows transforming liveness properties
into this fragment of safety properties~\cite{BaDrKr-2016-liveness}.
Their methodology is based on the idea that the protocol specifies
a way to reach the `desirable state,' e.g. by defining a recovery protocol.
Hence any trace either already reached a desirable state or it has not
exhausted all specified recovery steps --- which is a safety property.
Unfortunately, this approach does not apply here, as, in our case, the
protocol model is not meant to specify how an attack is
mounted.

An alternative approach to a direct encoding is to show that any trace
$t$ can be combined with any trace $t'$ that contains
$p_1,\ldots,p_n$, $c$, and nothing else. This may hold for processes
of a certain form. With such a result, protocol verifiers could again be
used to show the existence of $t'$. For the present paper, we leave
the verification of \autoref{asu:hyper} as an open question.
\end{discussion}

Under this condition, and if we assume the set of traces to be
prefix-closed, we obtain symbolic completeness.

\begin{theorem}
\label{the:symcomp}
    If
    \autoref{asu:postsing},
    \autoref{asu:prefixclosed},
    \autoref{asu:presubset}
    and
    \autoref{asu:hyper}
    hold, then
    $\PT$ is symbolically complete.
\end{theorem}
\begin{proof}
Induction over the length of $\tracewo{\pt}{\Inter}$.

\noindent
\emph{Base case:} $\abs{\tracewo{pt}{\Inter}}$ = $0$: Holds trivially
for $\sigma = \calI$.

\noindent
\emph{Inductive step}: Let 
$\tracewo{pt}{\Inter}$ = $e_1 ... e_k e_{k+1}$. 
From the IH and \autoref{asu:prefixclosed},
we know that there exists a trace $t_k \approx e_1,\ldots, e_{k} \approx \pt_k$.
By definition of $\approx$, we know that $e_{k+1}\in\Events$
and 
from $\pt\in\Ptraces$, we conclude that there is
$a_{e_{k+1}} \in \mathcal{A}$ with $e_{k+1}$ as a postcondition which 
was used to construct $\pt$.
By \autoref{asu:postsing},
$a_{e_{k+1}}$ = ($\pre$, $\{e_{k+1}\}$).
\revised{
By \autoref{asu:presubset} we know that $\pre \subseteq \Inter$.\\
The preconditions are met:
$\pre \subset \setfun{\tracewo{t_k}{\Inter}}$
because
$\pre \subset \setfun{\tracewo{\pt_k}{\Inter}}$.}
Thus, we can apply
\autoref{asu:hyper}
for $a= a_{e_{k+1}}$ and $t=t_k$
to obtain a trace
$t' \approx t_k \circ t_r$
with $\tracewo{t_r}{\Inter} = e_{k+1}$. 
Hence
$t' \approx t_k \circ e_{k+1} \approx e_1 ... e_k e_{k+1} \approx pt$.
\end{proof}

To summarize:
in conjunction,
Theorem~\ref{the:symsound}
and Theorem~\ref{the:symcomp}
ensure that 
the set of planning traces induced by $\PT$ and the set of protocol traces $\traceset$
are equal modulo $\Inter$ 
if
conditions \autoref{asu:postsing} --- \autoref{asu:postsplit} and  \autoref{asu:hyper} are met.
This is necessary for risk estimation techniques that  compute the
expected loss of value or the probability of a breach.

\autoref{asu:postsing} to \autoref{asu:postcomp} are satisfied w.l.o.g.\ for
monotonic planning tasks and \autoref{asu:prefixclosed} is a standard
assumption in protocol verification.
\revised{\autoref{asu:presubset} is a restriction we place on the composition of
the security model and auxiliary models for the planning task.}
The remaining assumptions \autoref{asu:postsplit} and
\autoref{asu:hyper} are both trace properties\revised{, the former
a safety property, the latter a liveness property.}
Given the lack of tool support, we will now focus on symbolic
soundness, which ensures that that the planning model considers all
possible attacks.
If symbolic soundness holds, any quantitative result that is monotonic
in $\calA$ --- e.g.  the expected damage or the probability of
reaching a critical asset --- can be considered an upper bound,
provided, of course, that 
model parameters such as the value of assets and probabilities of
actions are correct.

\section{Applications}
\label{sec:app}

The previous section results lay the foundation for using
highly optimized planners for the analysis of large networks. We
envision the following applications.

\paragraph{Protocol analysis for limited network attackers}

Today's protocol verification focuses on protocols in isolation \revised{and
against an}
attacker who can eavesdrop and modify all messages on
the network.
In terms of communication, this is the worst-case assumption for
distributed services on the Internet.
On the other hand, underlying services like the PKI or name resolution
are almost always trusted and, more often than not, vastly simplified
to the point of complete abstraction. 
For perspective, the Dolev-Yao model, which formalized these
assumptions, is older than the first
\revised{implementation} of name resolution. 
%

Planning models scale much better to large problem sizes (in terms of
actions) than protocol verifiers, and are thus able to analyze
the security of protocols in threat scenarios that are more
complicated to describe.  Incorporating more precise assumptions about
the attacker could lead to more nuanced results, e.g. about protocol
security in various topologies.

\paragraph{Cost-benefit guided protocol deployment in the Internet}

Deployment assessment techniques are based on an infrastructure threat
model and consider the deployment of a protocol as a ‘countermeasure.’
Using the recently proposed ‘Stackelberg planning algorithm,’ it is
possible to obtain the set of all Pareto-optimal protocol deployments
per node.
This allows for an evaluation of the actual benefit of new proposals
vis-à-vis the current infrastructure of the Internet.
It makes it possible to compare proposals against each other that are
incomparable on paper, e.g. is DNSSEC a better solution against
JavaScript injection attacks than application-specific techniques like
subresource integrity on the HTML level. 

This technique has been applied to email~\cite{speicher2018formally}
and the web~\cite{giorgio}, comparing solutions at the routing layer,
resolution layer and application layer.
A weak point of this methodology was the lack of justification for
their
attacker model. Symbolic soundness and completeness can bridge this gap,
as we will demonstrate for a subset of the email model~\cite{speicher2018formally}.
As we argued in Section~\ref{sec:back}, to justify the correctness
of the Stackelberg planning problem,
it is sufficient
to show the symbolic soundness/correctness for the attacker planning
problem, but for arbitrary initial states.

\paragraph{Corporate network analysis}

Risk assessment techniques for local networks (e.g.
mulVal~\cite{ou2005mulval}) focus on implementation-level flaws, e.g.
buffer overruns, but often ignore the protocol
level implications. An attacker that captures the company's certificate authority
or authentication server can usually exploit this
infrastructure's trust to obtain critical assets. Moreover, modern cloud-based
services introduce new dependencies on external infrastructure. These aspects are rarely considered and could be improved by
a rule-based representation of the involved protocol's flaws.


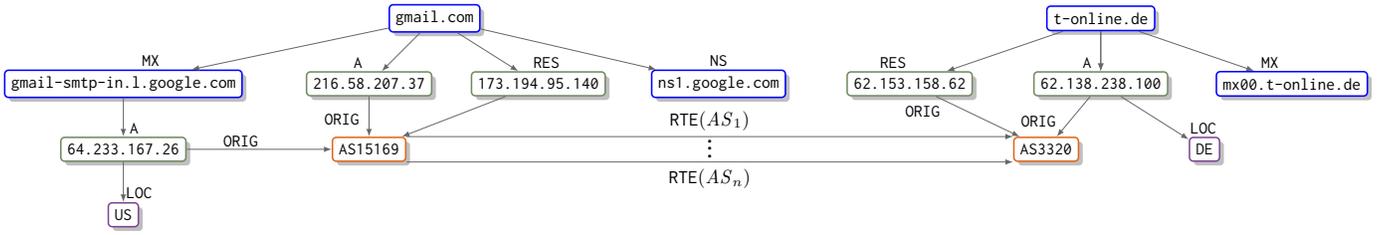
\begin{figure*}
\centering
\begin{adjustbox}{width=1\linewidth}
\begin{tikzpicture}[
  level 1/.style={sibling distance=38mm},    
  edge from parent/.style={->,draw=gray},
  level distance = 12mm,
  >=latex]

\node[dom] (root) {gmail.com} {

  child {node[dom] (d3) {gmail-smtp-in.l.google.com}
	child {node[ip]    (ip2) {64.233.167.26}
	child {node[cc]    (cc1) {US}}}}
  child {node[ip][right of=d3, xshift=3.5cm]   (ip1) {216.58.207.37}
	child {node[as]   (d6) {AS15169}
	}
}
  child {node[ip]   (ip_res) {173.194.95.140}}
  child {node[dom][xshift=-0.5cm]   (d1) {ns1.google.com}}

   
};

  \node[dom] (d12)  [right of=root, xshift=11.2cm] {t-online.de}{
    	 child {node[ip][right of=d1, xshift=2.5cm]   (ip_res2) {62.153.158.62}}
    	 child {node[ip][right of=ip_res2, xshift=2.5cm]   (ip3) {62.138.238.100}
        child {node[as][xshift=0.9cm]   (d7) {AS3320}}
       	 child {node[cc]    (cc2) {DE}}
}	
	child {node[dom][right of =ip3,xshift=2.5cm]   (d4) {mx00.t-online.de}    
}};
        
  \draw[->,draw=gray] (d12)  to (ip3);
  \draw[->,draw=gray] (ip_res)  to (d6);
  \draw[->,draw=gray] (ip2)  to (d6);
  \draw[->,draw=gray] (ip_res2)  to (d7);
    \draw[->,draw=gray] (d6.north east) -- (d7.north west) node[midway, above] {$\texttt{RTE}(\mathit{AS}_1)$} ;
    \draw[draw=none] (d6.east) -- (d7.west) node[midway, yshift=1.2mm] {\Large \vdots} ;
    \draw[->,draw=gray] (d6.south east) -- (d7.south west) node[midway, below] {$\texttt{RTE}(\mathit{AS}_n)$} ;
  \node [label, above=-0.06cm of d1,  xshift=0cm]   {\texttt{NS}};
  \node [label, above=-0.06cm of d3,  xshift=0.5cm]  {\texttt{MX}};
    \node [label, above=-0.06cm of d4,  xshift=-0.4cm]  {\texttt{MX}};
  \node [label, above=-0.06cm of ip1, xshift=-0.2cm]  {\texttt{A}};
  \node [label, above=-0.06cm of ip2, xshift=0.2cm]   {\texttt{A}};
    \node [label, above=-0.06cm of ip3, xshift=-0.25cm]   {\texttt{A}};
  \node [label, above=-0.06cm of cc1, xshift=0.29cm]   {\texttt{LOC}};
    \node [label, above=-0.06cm of cc2, xshift=-0.02cm]   {\texttt{LOC}};
  \node [label, above=0.1cm of d6,  xshift=-0.5cm]  {\texttt{ORIG}};
  \node [label, right=0.55cm of ip2,  yshift=0.15cm]  {\texttt{ORIG}};
    \node [label, above=0.05cm of d7, xshift=-0.14cm]  {\texttt{ORIG}};
     \node [label, below=0.05cm of ip_res2, xshift=0.24cm]  {\texttt{ORIG}};
      \node [label, above=-0.06cm of ip_res, xshift=0.15cm]   {\texttt{RES}};
            \node [label, above=-0.06cm of ip_res2, xshift=-0.3cm]   {\texttt{RES}};
\end{tikzpicture}
\end{adjustbox}
\caption{\revised{Snippet of the property graph.  (Taken from Figure 2 with permission of \cite{speicher2018formally}).}}
\label{fig:excerpt_G}
\end{figure*}

%
%
%

\section{Background: Email Case Study}
\label{sec:case}

We recall the email infrastructure attacker model
by~\citet{speicher2018formally} to justify its soundness in the next
chapter.
Using Stackelberg planning, they investigated how existing protocols
can be used to secure users against large-scale eavesdropping by
countries.
While the impact of many techniques is different depending on the
attacker and defender country (e.g. Russia and China are much more
self-reliant than, e.g. Brazil), the enforcement of TLS and improved
certificate validation have a significant impact throughout.
In the following, we will focus on their
threat model and infrastructure representation.


The email infrastructure is modeled as
a \emph{labeled property graph}~\cite{cispa1091}, \revised{which is simply
a graph with edge and node labels that describe
service providers and their interdependencies}.

\begin{definition}
\label{def:labeledpropgraph}
A \textit{labeled property graph} is a directed multigraph and described as a quadruple $G$ = ($V$,$E$,$\lambda$,$\mu$) over an alphabet $\Sigma$.
$V$ is the set of nodes. $E \subset (V \times V)$ is a set of edges between nodes.
The function $\lambda : V \cup E \rightarrow \Sigma$ maps a label from the alphabet $\Sigma$ to nodes and edges.
$\mu : (V \cup E) \times K \rightarrow S$ maps a string \emph{value} $s \in S$ to a node/edge and a \emph{key} $k \in K$.
\end{definition}

\begin{table}[ht]
\centering
\caption{Node labels (top) and edge labels (bottom).}
\label{table:labels}
\small
\begin{tabular}{ll}
\toprule
Labels & Description \\
\midrule
\addr		& Node for IP address \\
\domain		& Node for domain name.\\
\AS		& IANA number assigned to the AS.\\
\countr	& Country code \\
\midrule
\orig		& AS where lhs node originates from\\
\loc		& Country where lhs nodes is located\\
\A		& DNS record mapping Domain to Address\\
\MX		& DNS record mapping Domain to Domain\\
\NS		& DNS record for Name Servers\\
\DNS		& Resolving lhs requires resolving rhs \\
\RES	& lhs node uses resolver on rhs for
\		resolution\\
\RTE{AS$_t$}& AS-level route \revised{between ASes} via AS$_t$\\
\bottomrule
\end{tabular}
\end{table}

Table~\ref{table:labels} shows the node labels and edge labels 
used by~\citet{cispa1091}.
Figure~\ref{fig:excerpt_G} provides an example for the interaction
between two mail providers.
The green nodes represent IPv4 addresses and are labeled \addr.
They are associated to autonomous systems (orange, labeled $\AS$) via
the relation $\orig$.

The blue nodes represent domain names and are labeled $\domain$. They
are associated to IPs or other domains via the relations $\A$, 
$\MX$ and $\NS$ which encode the resources records that were obtained
by scanning. They designate the domain’s IP address ($\A$), its
responsible mail server ($\MX$) and its authoritative name server
($\NS$), respectively.

The label $\DNS$  records the relationship between authoritative
name servers and $\RES$ between mail servers and their resolvers.
$\RTE{\ASv_b}$ is used to record routing dependencies. If $\ASv_a$ is
connected to $\ASv_c$ and, somewhere along the way, a package might
traverse $\ASv_b$, an attacker at $\ASv_b$ could eavesdrop that
communication.
Domains, IPs and ASes are associated to countries via 
\loc edges.

\subsection{Infrastructure attacker model}\label{sec:threat}

An attacker in this model can be a country or a group of countries
that can corrupt all servers in their jurisdiction, as well as to
observe, intercept and alter all messages routed through their
jurisdiction.\footnote{Attacks from large ISP can be modeled
similarly~\cite{giorgio}.}
The IA model tracks infrastructure compromise at different levels
with corresponding predicates. 
For example, if the attacker has compromised the resolver of some
domain, we would consider the \emph{integrity} of all domain name
resolutions of this domain compromised. However, the resolved domains
themselves are not be compromised and may be safe to use for clients
that use a different resolver.
These predicates will be part of $\Inter$
and thus have to coincide with corresponding events in the protocol
model.

\begin{table}
\caption{Corruption predicates}
\label{table:corrupt}
\centering
\small
\begin{tabular}{lp{61mm}}
\hline
\compr($x$)			& Node $x \in  \domain \cup \addr \cup \AS \cup \countr$ under attacker control\\
\intrh($d$)	& Integrity of name resolution of $d \in \domain$ compromised\\
\intrrd($d',e'$)		& Integrity of some route from $d' \in \addr$ to $e' \in \addr$ is compromised\\
\intrh($d,e$)	& Integrity of name resolution of $e \in \domain$ from the perspective of $d \in \domain$ compromised\\
\unconf($d,e$)		& email communication from some user of $d \in$ Provider to some user of $e \in$ Provider
is considered unconfidential\\
\dnssec($d$)		& $d \in \domain$ does not support DNSSEC
\\
\hline
\end{tabular}
\end{table}

There are 16 rules (also called action schemas) that define how these
predicates can be derived. They are parametric in the graph: for
a given graph, they are compiled into a finite set of attacker actions
$\calA$ and predicates $\calP$.
Our focus is on the methodology; hence we will refer to
Appendix \ref{app:threat} for the full set of attacker rules and
only give a flavor of these rules with the following
simple example.

\begin{example}
\label{exa:attackrule}
\[
\infer{\intrd(d, e)}{d, e, r  \in \domain
& d \xrightarrow{\RES} r
& \intrrd(d, r) & \dnssec(e)}
\]
\end{example}

The intuition is as follows.
If the attacker controls the route from a domain to the resolver that
this domain uses, then
we consider the integrity of any name resolution this domain attempts
compromised. 
If the domain that is resolved uses DNSSEC, however, then the resolver
can verify the integrity of this signature and this attack vector is not available. 
(A different rule deals with the case where the resolver itself is
compromised.) 
%
The predicate $\dnssec$ cannot be produced by the attacker, as it 
is a defender predicate.

In \citeauthor{speicher2018formally}'s model, all attacker rules
that produce the predicate $\unconf$
are
associated with a reward in terms of the number of users affected. The
Stackelberg planning algorithm maximizes the sum of rewards.
As $\unconf\in\Inter$, the symbolic soundness result ensures that this
is an upper bound.

\subsection{Limitations}

To simplify presentation, we concentrate on the core model, consisting
of resolvers, DNS, DNSSEC and SMTP.
We left out SMTP over TLS, DANE and IPsec for secure inter-AS communication.
The protocol transformations we present in the next section would
apply to the full model, as these protocols could be added without
changing the structure of the processes.
Thus the methodology would be the same, but the ProVerif processes
would need to be extended.

The attacker model is not probabilistic, but relies on correct
attacker rewards, which \citet{speicher2018formally} estimated from
public sources.
These are model parameters and
need to be estimated.

\section{Background: ProVerif}
\label{sec:process}

In the following, we introduce ProVerif's dialect of the applied$-\pi$
calculus~\cite{blanchet_efficient_2001,blanchet2018proverif}. Readers familiar with it can safely skip to the next
section.

\subsection{Syntax}

\begin{figure}[ht]
\centering
\begin{tabular}{ll}
$M, N ::=$ & terms \\
\ \ $v,x,y,z$ & \ \ variable \\
\ \ $a,b,c,n,k,s$ & \ \ free name \\
\ \ $p,q$ & \ \ public name \\
\ \ $f(M_1,...,M_n)$ & \ \ constructor application \\
\ & \ \\
$P,Q ::=$ & processes \\
\ \ 0 & \ \ nil \\
\ \ $P|Q$  & \ \ parallel\\
\ \ $!P$ & \ \ replication\\
\ \ $\pin(M,x) .P$ & \ \ input\\
\ \ $\pout(M,N) . P$ & \ \ output\\
\ \ $(\nu a)P$ & \ \ restriction\\
\ \ $\pif\ M = N\ \pthen\ P$ & \ \\
\ \ \ \ \ \ \ \ \ \ $\pelse\ Q$ & \ \ conditional \\
\ \ $\plet\  x = g(M_1,...,M_n)\ \pin$  & \ \\
\ \ \ \ \ \ \ \ \ \ $P\ \pelse\ Q$ & \ \ destructor application\\
\ \ $\pevent(M) . P$ & \ \ event\\
\end{tabular}
\caption{Syntax of the process calculus}
\label{syntax}
\end{figure}

We present the syntax of the calculus in Figure~\ref{syntax}. Terms represent messages
and data. Processes represent entities/programs. We use $x,y,z$ to represent
variables, $a,b,c,n,k,s$ for free names and $p, q$ for public names.
We use $\FN$ and $\PN$ to refer to the set of free names and public names, respectively.
Both are arbitrary, but infinite.
The function symbol $f$ represents a \textit{constructor}
whereas we use $g$ to represent \textit{destructors}. 
Both are abstract function symbols with some fixed arity. 

\textit{Terms} are defined over names, variables, and the applications
of constructors. 
Destructors are used to manipulate terms in processes:
$let\ x = g(M_1,...,M_n)$ $in\ P$ $\pelse\ Q$ binds x to the result of the 
destructor application of $g$ on $M_1...M_n$ and continues with process
$P$. If the application fails, however, we continue with process $Q$.
A destructor g is defined by a finite set of 
reductions $\mathit{def}(g) := g(N_1,...,N_n) \rightarrow N$ where the terms
$N,N_1,...,N_n$ are build without free names and $var(N) \subset 
var(N_1 \cup ... \cup N_N)$. A destructor fails, if no reduction
applies.

\begin{example}
Symmetric encryption is described by a 2-ary constructor 
$\senc$ and a 2-ary destructor with the following reduction:
\[
\sdec(\senc(x,y),y) \to x 
\]
\end{example}

We write $fn(P)$ (and $fv(P)$) for
the sets of names (variables) that are free in $P$.
A substitution $\Frame$ = $\{ \substi{t_1}{x_1},$ ..., $\substi{t_n}{x_n} \}$ 
is a partial function, mapping variables to terms. The domain of $\Frame$ is
$\mathbb{D}(\Frame)$ = $\{x_1, ..., x_n\}$ and $\Frame$ maps $x_i$ to $t_i$.
The application of $g$ on the terms
$M_1,...,M_n$ is defined if and only if there exists some substitution
$\Frame$ and a reduction rule $g(N_1,...,N_n) \rightarrow N$ such
that for all $i \in \{1,...,n\}$ it holds that $N_i = M_i\sigma$.
In this case, $\plet\ x = g(M_1,...,M_n)\ \pin\ P\ \pelse\ Q$ would bind x
to $N\Frame$ and continue to execute $P$. 

Additionally, the process calculus provides the instruction $\pevent(F).P$ to emit
some $F \in \EventSig$ as an annotation of the process and continue to execute $P$.
We define the set of these annotations as
\[
    \Events := \set{F(t_1, .., t_k) \mid t_i\text{ terms}, F\in\EventSig \text{ with arity $k$}}. 
\]
%
The remaining constructs depicted in Figure~\ref{syntax} are standard
constructs included in the $\pi$-calculus. $0$, or the nil process,
indicates the end of the process and does nothing. $P|Q$ composes $P$
and $Q$ in parallel and $!P$ represents and unbounded number of copies
of $P$ in parallel composition. A \textit{channel} can be any term $M$.
The process $\pin(M,x).P$ receives a message on channel $M$. It then continues to execute $P$ with $x$ being bound to the received message. 
$\pout(M,N) . P$ outputs a term $N$ on channel $M$
and executes $P$. $(\nu a)P$ depicts a restriction. It first creates
a free name $a$ and then executes $P$. A free name $a$ is a secret
and cannot be guessed, but it may be obtainable via computation/deduction of
public messages. 
The conditional $\if\ M = N\ \then\ P\ \else\ Q$ compares
two terms $M$ and $N$ and executes $P$ if they are equal and $Q$ otherwise.
\begin{full}
Note that this is just a special case of destructor application. Let $\mathfun{equal}$
be a destructor symbol and $\mathfun{def(equal)}$ = $\set{\mathfun{equal}(M,M) \rightarrow M}$.
Then $\pif\ M = N\ \pthen\ P\ \pelse\ Q$ can also be expressed as $\let\ x=\mathfun{equal}(M,N)$
$\pin\ P\ \pelse\ Q$ with $x$ being not free in $P$ and $Q$.
\end{full}
For brevity, we will omit trailing $0$ processes and empty else-branches.
%

\begin{figure*}
\begin{align*}
(\Eps, \Process \mulcup \{0\}, \Frame) \ \ &\rightarrow \ \ (\Eps, \Process, \Frame)\tag{\textsc{null}}\\
(\Eps, \Process \mulcup \{P \mid Q\}, \Frame) \ \ &\rightarrow \ \ (\Eps, \Process \mulcup \{P,Q\}, \Frame)\tag{\textsc{par}}\\
(\Eps, \Process \mulcup \{!P\}, \Frame) \ \ &\rightarrow \ \ (\Eps, \Process \mulcup \{P,!P\}, \Frame)\tag{\textsc{repl}}\\
(\Eps, \Process \mulcup \{\nu a;P\}, \Frame) \ \ &\rightarrow \ \ (\Eps \cup \set{b}, \Process \mulcup \{P\set{\substi{b}{a}}\}, \Frame)\tag{\textsc{new}} \ \text{if b is free an not in }\Eps\\
(\Eps, \Process \mulcup \{\text{out}(t,M);P\}, \Frame) \ \ &\rightarrow \ \ (\Eps, \Process \mulcup \{P\}, \Frame \cup \{\set{\substi{M}{x}}\})\tag{\textsc{out}} \ \text{if x is fresh and}\ \nu\Eps.\Frame \vdash t\\
(\Eps, \Process \mulcup \{\text{in}(t,x);P\}, \Frame) \ \ &\rightarrow \ \ (\Eps, \Process \mulcup \{P\set{\substi{M}{x}}\}, \Frame)\tag{\textsc{in}} \ \text{if}\ \nu\Eps.\Frame \vdash M\ \text{and if}\ \nu\Eps.\Frame \vdash t\\
(\Eps, \Process \mulcup \{\mathit{let}\ x=M\ \mathit{in}\ P\ \pelse\ Q\}, \Frame) \ \ &\rightarrow \ \ (\Eps, \Process \mulcup \{P\set{\substi{M}{x}}\}, \Frame)\tag{\textsc{lets}} \ \text{if evaluation of M succeeds}\\
(\Eps, \Process \mulcup \{\mathit{let}\ x=M\ \mathit{in}\ P\ \pelse\ Q\}, \Frame) \ \ &\rightarrow \ \ (\Eps, \Process \mulcup \{Q\}, \Frame)\tag{\textsc{letf}} \ \text{if evaluation of M fails}\\
(\Eps, \Process \mulcup \{event(F);P\}, \Frame) \ \ &\xrightarrow{F} \ \ (\Eps, \Process \mulcup \{P\}, \Frame)\tag{\textsc{event}}
\end{align*}
\caption{Operational semantics.}
\caption*{Note that evaluation of some $M$ succeeds, if for all destructor symbols in $M$, there is an applicable rewrite rule. If there is a destructor symbol in $M$ which has no applicable rewrite rule, then evaluation fails.}
\label{semantics}
\end{figure*}

\subsection{Semantics}
\label{semanticsapp}

\begin{figure}[ht]
\centering
\[
    \infer[(\textsc{DName})]{\frameabr \vdash a}{a \in FN \cup PN\ a \not\in \overrightarrow{n}},\,
    \infer[(\textsc{DFrame})]{\frameabr \vdash x\Frame}{x \in \mathbb{D}(\Frame)}
\]
\[
    \infer[(\textsc{DCon})]{\frameabr \vdash f(t_1, ..., t_n)}{\frameabr \vdash t_1 & \ldots & \frameabr \vdash t_n\ & f\ \not\in\ \priv}
\]
\[
    \infer[(\textsc{DDes})]{\frameabr \vdash t}{\frameabr \vdash t_1\ ...\ \frameabr \vdash t_n & \{d(t_1,...,t_n) \rightarrow t\} \in \mathit{def}(g)}
\]
\caption{Deduction rules.}
\label{deducrules}
\end{figure}

We define the semantics 
by first introducing the notions of \textit{frame} and \textit{deduction}.
A frame $\frameabr$ represents
a sequence of messages observed so far and the secrets
generated by the protocol. The first is captured by
a substitution $\Frame$, the latter by the set of used names $\Eps$.

Deduction describes the capabilities of an 
adversary to infer and compute new terms from already observed messages.
We define the deduction relation $\frameabr \vdash t$ 
between a frame and a derivable term as the smallest relation s.t.\
the rules in Figure \ref{deducrules} hold.
We further define $\priv$ as a subset of all constructor symbols where
the \textsc{DCON} deduction rule cannot be used. We refer 
to $\priv$ as \textit{private constructor symbols}.

The \textit{operational semantics} are defined by a labeled transition
relation between process configurations. This \textit{configuration}
is represented by a 3-tuple $(\Eps,\Process,\Frame)$. 
$\Process$ is a multiset representation of processes being executed in parallel.
$\Eps$ is the set of free names generated by the processes in $\Process$.
$\Frame$ is a substitution modeling the messages observed by the environment.

The labeled transition relation of our calculus can be found in Figure~\ref{semantics}.
Each transition between two configurations is labeled with some $F \in \Events \cup \{\emptyset\}$.
For the ease of presentation, we omit empty sets and write $\rightarrow$ instead of $\xrightarrow{\emptyset}$.
We define $\rightarrow^{*}$ to represent multiple application of transition rules labeled with
the empty set. For the other $F \in \Events$ we define $\xRightarrow{F}$ as
$\rightarrow^* \xrightarrow{F} \rightarrow^*$.

\begin{definition}[Traces]
\label{def:trace}
Given a process $P$, we now define its traces:
\begin{multline*}
\traces(P) = \left\{ ( F_1,...F_n ) \mid (\emptyset, \{P\}, \emptyset)
\xRightarrow{F_1}(\Eps_1, \Process_1, \Frame_1) \right. \\
\left.
\xRightarrow{F_2} \ldots
\xRightarrow{F_n}(\Eps_n, \Process_n, \Frame_n) \right\}
\end{multline*}
\end{definition}

\section{Case Study: Email} 
\label{sec:caseres}

We now come back to \citeauthor{speicher2018formally}'s email case study (Section~\ref{sec:case}) to investigate the symbolic soundness of their model.
Our focus will be on the methodology. We first present a translation from labeled property graphs into processes. Verifying this process for each property graph is impractical, both because of the size of the graph (protocol verifiers do not scale well with the model size) and 
because any change in the infrastructure would require a new analysis.
Hence, we define two process transformations that allow for a sound mapping of all these processes to a single process, i.e. an over-approximation. We verify this process in ProVerif and can thus provide a symbolic soundness result for all process graphs at once.




\subsection{Symbolic model}
\label{sec:models}
\label{subsec:cons}

\lstset{
  language={Proverif},
  basicstyle=\small\ttfamily, 
  columns=fixed, 
  keepspaces=true, 
  showstringspaces=false, 
  breaklines=true, 
  numbers=none
}

\begin{figure*}[ht!]
    \small
\centering

\begin{align*}
    \calF(G)=
\bigPar\limits_{v \in V_\prov} 
\left(
    \begin{array}{lcl}
 & \bigPar\limits_{\substack{
    d_\MX^\client, i_\MX^\client, \as_\MX^\client,
    e_\MX^\client, j_\MX^\client, \as_1 \in V_\MX,
    r^\client_\RES, j_\RES^\client, \as_2 \in V_\RES, 
    \as^1 , \as^2 \in V, 
    v_2 \in V_\prov
    .
    \\
    v \rightarrow d_\MX^\client \xrightarrow{A} i_\MX^\client
    \xrightarrow\orig \as_\MX^\client,
     \\
     v_2 \rightarrow e_\MX^\client \xrightarrow{A} j_\MX^\client \xrightarrow\orig \as_1
     \wedge (\as_1 = \as_\MX^\client \vee \as_1 \xrightarrow{\RTE{\as^1}} \as_\MX^\client),
 \\
 r^\client_\RES \xrightarrow{A} j_\RES^\client \xrightarrow\orig \as_2
 \wedge (\as_2 = \as_\MX^\client \vee \as_2 \xrightarrow{\RTE{\as^2}} \as_\MX^\client),
\\
\vecin{c} = (v, d_\MX^\client, i_\MX^\client, \as_\MX^\client) ,  \vecin{s} = (v_2, e_\MX^\client, j_\MX^\client, \as_1), \vecin{r} = ( r^\client_\RES, j_\RES^\client, \as_2)
    }}
&
    !P_\smtpclient(\vecin{c},\vecin{s},\vecin{r})
\\
\mid 
&
\bigPar\limits_{\substack{
        d_\MX^\server, i_\MX^\server, \as_\MX^\server,
		e_\MX^\server, j_\MX^\server, \as_1 \in V_\MX, 
		\as \in V,
		v_2 \in V_\prov
        .\\
v \rightarrow d_\MX^\server \xrightarrow{A} i_\MX^\server
\xrightarrow\orig \as_\MX^\server 
\\
 v_2 \rightarrow e_\MX^\server \xrightarrow{A} j_\MX^\server \xrightarrow\orig \as_1
     \wedge (\as_1 = \as_\MX^\server \vee \as_1 \xrightarrow{\RTE{\as}} \as_\MX^\server),
     \\
     \vecin{s} = (v, d_\MX^\server, i_\MX^\server, \as_\MX^\server) ,  \vecin{c} = (v_2, e_\MX^\server, j_\MX^\server, \as_1)
}}
&
!P_\smtpserver(\vecin{s}, \vecin{c})
\\
\mid
&
\bigPar\limits_{\substack{
        d_{\RES}, i_{\RES}, \as_{\RES} \in V_\RES
	  e_{\MX}, i_\MX, \as_1 \in V_\MX,
	n_{\DNS}, i_\DNS, \as_2 \in V_\DNS, 
	n_{\rns}, i_\rns, \as_3 \in V_\rns, 
	\as^1, \as^2, \as^3 \in V
        .\\
v \rightarrow d_{\RES} \xrightarrow{A} i_{\RES}
 \xrightarrow\orig \as_{\RES} ,
  \\
  (v \rightarrow e_{\MX} \xrightarrow{A} i_\MX \xrightarrow\orig \as_1,
  \wedge (\as_1 = \as_\RES \vee \as_1 \xrightarrow{\RTE{\as^1}} \as_\RES),
  \\
(n_{\DNS} \xrightarrow{A} i_\DNS \xrightarrow\orig \as_2,
 \wedge (\as_2 = \as_\RES \vee \as_2 \xrightarrow{\RTE{\as^2}} \as_\RES),
 \\
(n_{\rns} \xrightarrow{A} i_\rns \xrightarrow\orig \as_3,
 \wedge (\as_3 = \as_\RES \vee \as_3 \xrightarrow{\RTE{\as^3}} \as_\RES),
\\
\vecin{r} = (d_{\RES}, i_{\RES}, \as_{\RES}), \vecin{c} = (v, e_\MX, i_\MX, \as_1),  \vecin{d} = (d_{\DNS}, i_{\DNS}, \as_2), \vecin{root} = (d_{\rns}, i_{\rns}, \as_3)
  }} 
&
!P_\dnsres(\vecin{r}, \vecin{c}, \vecin{d}, \vecin{root})
\\
\mid 
&
\bigPar\limits_{\substack{
        d_{\DNS}, i_{\DNS}, \as_{\DNS}  \in V_\DNS,
	r_{\RES}, i_\RES, \as_1 \in V_\RES, 
	\as \in V
        .\\
 v \rightarrow d_{\DNS} \xrightarrow{A} i_{\DNS}
  \xrightarrow\orig \as_{\DNS},
  \\
  r_{\RES} \xrightarrow{A} i_\RES \xrightarrow\orig \as_1
  \wedge (\as_1 = \as_\DNS \vee \as_1 \xrightarrow{\RTE{\as}} \as_{\DNS} ),
  \\
\vecin{d} = (d_{\DNS}, i_{\DNS}, \as_{\DNS}), \vecin{r} = (r_{\RES}, i_\RES, \as_1)
  }} 
&
 !P_\dnsns(\vecin{d}, \vecin{r} ) 
 \\
 \mid 
&
\bigPar\limits_{\substack{
        d_{\rns}, i_{\rns}, \as_{\rns}  \in V_\rns,
r_{\RES}, i_\RES, \as_1 \in V_\RES, 
 \as \in V
        .\\
 v \rightarrow d_{\rns} \xrightarrow{A} i_{\rns}
  \xrightarrow\orig \as_{\rns},
  \\
   r_{\RES} \xrightarrow{A} i_\RES \xrightarrow\orig \as_1
  \wedge (\as_1 = \as_\DNS \vee \as_1 \xrightarrow{\RTE{\as}} \as_{\rns} )
   \\
  \vecin{root} = (d_{\rns}, i_{\rns}, \as_{\rns}), \vecin{r} = (r_{\RES}, i_\RES, \as_1)
  }} 
&
 !P_\dnsrns(\vecin{root}, \vecin{r}  ) 
    \end{array}
\right)
\end{align*}
\caption{Function $\mathcal{F}$ from property graphs to processes.}
\label{paraprocess}
\end{figure*}

We define a function $\calF$ from property graphs to processes in
Figure~\ref{paraprocess}. We use the following notation:
\begin{itemize}
\item For a finite set $S=\set{a,..,z}$, 
    $\bigPar\limits_{s \in S} P(s)$ denotes  $P(a) \mid \ldots \mid P(z)$.
    Instead of $s\in S$, we sometimes use set-builder notation to
    directly define the components of each $s$. 
\item For a fixed labeled property graph $G$ that is implicit in the
    context, we write $V_x$ as a subset of all nodes in $V$ with label
    $x$. We write $y \xrightarrow{L} z$ to represent an edge in
    G labelled with $L$ connecting the two nodes $x$ and $y$. To ease
    notation, we use $d,e$ for nodes representing domain names, $r$
    for resolvers and $n$ for name servers.
\item We further assume that all nodes are public names, to avoid
    introducing a mapping.
\end{itemize}

Our process represents SMTP, DNS, DNSSEC, resolvers, and a simplified version of inter-AS communication. As we focus on the methodology, we do not elaborate on the subprocesses modeling these protocols, but on the top-level process that composes them.
The processes $P_\smtpserver$ and $P_\smtpclient$ describe the client and server roles within the SMTP protocol. Each provider $v$ defines several mail servers  $d_\MX^{\client / \server}$ (or $e_\MX^{\client / \server}$) via the $\MX$ resource record. Each of those execute both client and server roles. They have one or many
IP addresses  $i_\MX^{\client / \server}$, which are located in autonomous systems $\as_\MX^{\client / \server}$. Whereas the process $P_\smtpserver$ only models the receiving part of the SMTP protocol, $P_\smtpclient$ models DNS/DNSSEC requests as well as the client role of SMTP.
To establish the connections between the different services, we use the IP addresses to model channels between them. These channels are built over private constructors. 
In contrast to the Dolev-Yao model,  the attacker cannot eavesdrop or manipulate messages per default,  but needs to obtain access to these channels by compromising either domain names, IP addresses, or ASes.

The processes $P_\dnsres$, $P_\dnsns$, and $P_\dnsrns$ describe the resolver and server role within the DNS protocol, which, depending on the server's configuration, include the DNSSEC extension. Process $P_\dnsres$ models the resolver role by communicating with the DNS/DNSSEC infrastructure, on the one hand, and with the requesting role of the mail server. As with the previously mentioned process, the IPs are used to construct private channels via private constructors.
The same holds for the name server role modeled by $P_\dnsns$. An exception is the process $P_\dnsrns$ modeling root name servers. We assume that root servers cannot be corrupted, since that would break the DNS/DNSSEC infrastructure as a whole. Therefore, the attacker is not able to corrupt the connection between the root server process and the process modeling the resolver role. Connections established by the processes $P_\dnsres$ and $P_\dnsns$, however, may be corrupted by corrupting their domain names, IP addresses, or ASes. 
For simplicity, we constrained our DNS model to two levels of name servers.
%

With this construction we represent the structure of the IA model.
We instantiate all communication paths and relations in the IA model using the same labeled property graph $G$. Further, all featured protocols and functionalities of the IA model are represented by subprocesses in our model, as well as the notion of corruption.

In the follow up, we will modify the top-level structure, but leave
the processes 
$P_\smtpclient,
P_\smtpserver,
P_\dnsres,
P_\dnsns$
and
$P_\dnsrns$
intact. They are detailed in \appendixorfull{sec:full-model}.

\subsection{Proof via sound process transformations}
\label{justify}

With $\calF$, we can, in principle, verify the symbolic soundness
of each planning task induced by some property graph $G$.
The respective model $\calF(G)$ can become very large: property graphs
can have thousands to millions of nodes, whereas the majority of
protocol models fits on a piece of paper. Protocol verifiers are not
optimized for models of this size.
%
Moreover, it is tedious to generate and verify a process whenever a new
attacker country is considered or the property graph is modified.
Last but not least, the analogy to computational soundness (Sec.~\ref{sec:analogy})
suggest that symbolic soundness results
(a) should encompass some set of protocols 
and
(b) apply to any network composed of them, here described by the property graph.
\footnote{Computational soundness results fix a set of
cryptographic primitives, but hold for a class of protocols.}
Conceptually, we therefore desire a result that is independent of $G$.

\revised{
To this end, we propose the following proof technique specific to the
applied-$\pi$ calculus.
Let $\calF(G)$ be a function from property graphs\footnote{We use
property graphs for concreteness, the actual representation of
the network is irrelevant, as long as it translates to planning models
and processes in a uniform way.}
to processes and assume that it can be expressed only using 
the applied-$\pi$ calculus and
the meta language operation $\bigPar\limits_{s \in S} P(s)$.
In the first step, we
construct a process $P$ such that, for all $G$,
$\traces(\calF(G)) \subseteq \traces(P)$. 
This implies that every trace property that holds for $P$ also holds for
$\calF(G)$, independent of $G$.
To this end, we apply to two transformations 
that
over-approximates a process.
}

\revised{
The first permits substituting
several uniform parallel processes 
$\bigPar\limits_{s \in S} P(s)$
by a single
process under replication that obtains this input from the adversary.
In the description of $\calF(G)$, $G$ can only occur within these $S$, hence the
resulting process is now independent of $G$. 
}
\begin{lemma} With
\label{con:replin}
\begin{align*}
\traces(!\pin(\vecv).P) \supseteq \traces(\bigPar\limits_{\vecin{p}}P\substi{\vecin{p}}{\vecv})
\end{align*}
we relate the replication of a process $in(\vecv).P$, where $\vecv$ can be supplied 
by the adversary (i.e., the frame), to a finite parallel execution of the same process $P$, where $\vecv$
gets substituted with public names supplied by $G$. 
\end{lemma}

\revised{
    We may now automatically conduct the following verification steps 
 with ProVerif (or any other
verifier). ProVerif's abstraction in particular is sensitive to how
deep in a process an input occurs. The second transformation thus
permits pushing inputs further inside the process to aid
verification.
}
\begin{lemma}
\label{lem:in}
For all processes $Q$ that contain exactly one subprocess $\pin(x).P$,
let $Q'$ be $Q$ with $P$ instead of this subprocess. Then:
$
\traces(\pin(x).Q')
\subseteq
\traces(Q).
$
\end{lemma}
\revised{
Both lemmas are proven in Appendix~\ref{appendlemma}.
In the second step, we ensure that these transformations have been
correctly applied, i.e., that
$\traces(\calF(G)) \subseteq \traces(P)$
follows from Lemma~\ref{con:replin}~and~\ref{lem:in}.
In the third step, we verify the three
syntactic conditions on the planning task.
Finally, we use ProVerif 
to show
\revised{Condition~\autoref{asu:postsplit}}.
We partition  \calA by postcondition. For this condition to hold, each
element must have the following form.\revised{
\[
    \mathcal{A}_i =
    \{ (\pre_i^1, \{\post_i\}),(\pre_{i}^2, \{ \post_i\}),...,(\pre_i^k,
\{ \post_i \})\}.
\]}
We translate each element into a correspondence query: for each trace, if $\post_i$ occurs 
then 
$\pre_i^1$,
$\pre_i^2$ or any other 
$\pre_i^j$, $j \in \{1,\cdots,k\}$, occurs as well.
As Condition~\autoref{asu:prefixclosed} holds for any
ProVerif process, we
obtain symbolic soundness by Theorem~\ref{the:symsound}.
}

\subsection{ProVerif Model Introduction}
\label{sec:proverif}

Using the two lemmas above, we can transform \emph{all} $F(G)$ into the
process $P$, whose structure we present
in 
Figure~\ref{pshort}.
The full model is in \appendixorfull{sec:full-model}.
\begin{figure}[ht!]
\begin{lstlisting}[mathescape=true,breaklines=true,numbers=none,language=Proverif,basicstyle=\footnotesize]
!(in(c,prov:provider);
  !(in( dom_c:dom); in( ip_c:ip); in( AS_c:as);
    !($P_\smtpclient$(prov,dom_c,ip_c,AS_c)))
| !(in( dom_s:dom); in( ip_s:ip); in( AS_s:as);
    !($P_\smtpserver$(prov,dom_s,ip_s,AS_s)))
| !(in( dom_r:dom); in( ip_r:ip); in( AS_r:as);
    !($P_\dnsres$(dom_r,ip_r,AS_r)))
| !(in( dom_d:dom); in( ip_d:ip); in( AS_d:as);
    !($P_\dnsns$(dom_d,ip_d,AS_d)))
| !(in( dom_rn:dom); in( ip_rn:ip); in( AS_rn:as);
    !($P_\dnsrns$(dom_rn,ip_rn,AS_rn))))
\end{lstlisting}
\caption{Simplified ProVerif model ($\pin(m)$ short for
$\pin(c,m)$).}
\label{pshort}
\end{figure}

In this model, the adversary is also able to choose which processes
communicate, and thus controls the underlying network topology.
%
Given the process model based on the graph $G$ in Figure~\ref{paraprocess} and our ProVerif model
$P$, we show that the transformations have been applied correctly.
\begin{theorem}
\label{the:main}
$\forall G.
\traces(\mathcal{F}(G))
\subseteq
\traces(P) 
$.
\end{theorem}
\revised{
Lemmas~\ref{con:replin}~and~\ref{lem:in}
reduce the proof of this theorem to a structural argument 
(see Appendix~\ref{appendlemma}).
The syntactic conditions on the planning task can be verified 
by inspecting it (Appendix~\ref{app:threat}). 
\revised{Condition~\autoref{asu:postsing}},
discussed in Sec.~\ref{sec:defsymsound},
holds as there is no negated postconditions.
\revised{Condition~\autoref{asu:postcomp}} holds, as all events in
$\Inter$ occur as a postcondition of some rule.
\revised{Condition~\autoref{asu:presubset}} holds, as all
preconditions are in
$\Inter$. 
}

\subsection{Automated verification}

It remains to show \revised{Condition~\autoref{asu:postsplit}}, which
we verify using ProVerif.
The full set of queries is specified in~\appendixorfull{queries}. 
We grouped these queries according to whether the postcondition 
expresses a loss of integrity 
or 
a loss of confidentiality.
We express the first property as a correspondence property
and the second as a reachability property known as \emph{weak secrecy}.
%
%
For the first kind, we verify that any event matching some pattern $e$
was preceded by an event matching a pattern $e'$.  
Any trace with an event matching $e$ but not $e'$ could be mapped to
one where such integrity violation events are specifically marked;
these would be the actual events in $\Inter$.
For the second kind, weak secrecy is expressed as usual.
The attacker can demonstrate the ability to correctly input a secret message (in this case, the
content of an email) in a subprocess. Upon success, the subprocess
can be reduced to an event. We analyze the reachability of this event.

During the modeling process, we found two bugs in the IA model.
%
%
First, in the IA model,
$\compr(\ip)$ implies $\compr(d)$ if $d$ resolves to $\ip$,
but not vice versa.
As $\compr(ip)$ does not represent IP-level attacks, but a compromise
of the service identified by $\ip$, this ought to be the case.
Without this rule, all rules that concern routing, name resolution or
application compromise break down, as they identify the service with
the domain it runs on.
Luckily, this does not invalidate \citeauthor{speicher2018formally}'s
result, as an inconsistency between the corruption of an IP and
a domain can only come from (a) missing or inconsistent information in
the property graph, e.g. domains $d_1$ and $d_2$ linked to different
countries but resolving to the same IP, or (b) from an inconsistent
initial network attacker state. 
We confirmed with the authors that neither condition was met.

The second bug concerns the DNSSEC protocol. DNSSEC 
requires resolver-side signature validation. This is not always the
case for
resolvers run by ISPs, but a realistic future scenario to investigate. 
By contrast,
local resolvers (e.g. on clients or services like mail) 
rely on the ISP's validation during (recursive) resolving and will, at
least for the near future, 
not
validate signatures themselves.
\nameref{dnsrouteres}, 
however, assumes that this is the case, i.e. 
that DNSSEC is an effective countermeasure against
domain poisoning attacks mounted 
between the local (recursive, usually non-validating) resolver
and the ISP's (iterative, validating) resolver. 
Presumably, this is a bug, or at least an unrealistic
assumption.

To solve the first problem, we added a rule turning $\compr(d)$ into $\compr(\ip)$
and use $\compr(\ip)$ in all other rules, instead of $\compr(d)$.
To solve the second problem, we altered the rule by deleting the
$\dnssec$ predicate from the precondition. The changes are highlighted
in Appendix~\ref{app:threat}.

%
ProVerif proves all queries automatically and thus 
the last condition,~\autoref{asu:postsplit}.
We used ProVerif version 2.02pl1. On an Intel i7-9750H CPU with 16 GB
RAM, the analysis took 9.92s.
This concludes our proof for the symbolic soundness.

\subsection{Modeling Challenges}\label{discuss}

We take the opportunity to discuss some modeling challenges that we
encountered and that are specific to our methodology.

The first is the modeling of the infrastructure attacker, who is less
powerful than ProVerif's standard network attacker. It can only
observe communication if the corresponding route has been compromised.
Our first approach used \textit{private channels}
to model non-corrupted transfer of messages. We noticed ProVerif
running into termination problems during the resolution. We minimized
our model to three parties and found the issue 
to be ProVerif's internal representation of private channels as Horn clauses.
This is because private channels are synchronous, as opposed to free (public)
channels.%
\footnote{
In addition, \citeauthor{DBLP:journals/jcs/BabelCK20} point out various
communication semantics.}
Routing in the Internet is actually asynchronous, so we model secret
channels using 2-ary fact symbols \reqchannel, \anschannel and the following
reduction:
\[
\getreqpacket(x,\reqchannel(x,y))=y.
\]
(The reductions for \anschannel are analogous).
All parties apply the function symbol with a shared key in the first
parameter, to represent communication on that channel.
The keys are built over names representing the IP addresses
of the communicating parties, as well as a freshly chosen
source port (sender) and the publicly known target port.
To corrupt a key, the adversary claims the entity as part of
its domain.
Additionally, the adversary may choose some AS under its control
and claims it to be part of the IP route between the communicating parties.
With the corruption of this AS, the route is also seen as corrupted and
the adversary can claim the key. This in-transit AS corruption model
is very similar to the threat model described in the IA model.

The second challenge is how to structure the process such that
information about corruption at the routing level is transmitted to
processes that represent the resolution or application layer.
As an example, imagine the adversary compromising an AS. All service
providers affected would need to be informed that they can now output
their keys. First attempts with private channels 
lead to non-termination. Instead, we restructured the
process so that an AS compromise is a subprocess of the lower
layer. Each entity needs to be compromised separately, but raises the
same event $\compr(\as)$.

The third challenge is the size of the model. Using ProVerif's
pretty printing, the process counts 360 lines, which is unusually
large. The queries alone take about 47 lines. 
The most recent ProVerif release 2.02pl1 improved the verification
time from 4 minutes (with 2.01) to about ten seconds. Hence we do not see a reason why
the model could not be extended to \revised{cover \citeauthor{speicher2018formally} 's} complete model at a reasonable level
of abstraction.
Nevertheless, a full-blown model of TLS could bring ProVerif to
its limits.
We suspect the model size is the reason why resolution takes
unusually long --- typically, ProVerif's analysis takes seconds or
does not terminate at all.
Disabling either the DNSSEC or DNS
processes supports our suspicion that,
the model size has
a strong impact on the verification time, 
even though the models are
relatively simple.
A potential remedy is techniques for vertical and parallel
composition (e.g. \cite{escobar2010sequential},
\cite{cheval2017secure},
\cite{hess2018stateful}), which
could potentially be used to derive conditions for
the composition of IA models.


\section{Conclusion}
\label{sec:con}

We introduced the first formal approach to justify vulnerability analysis and risk assessment techniques that operate on an Internet-wide scale. We provided a formal methodology to analyze a given model with off-the-shelf verifiers and demonstrated the applicability of our approach for symbolic soundness w.r.t.\ a real-world IA model. The protocol transformations and modeling tricks to represent infrastructure attackers in the Dolev-Yao model might be of independent interest for protocol analysis, e.g. for the analysis of p2p protocols.

We identify two main limitations: first, the verification of symbolic completeness requires either true\footnote{See discussion in Sec.~\ref{sec:asusymcomp}.} \revised{support for liveness} in existing verifiers, or syntactical conditions that ensure that liveness can be concluded from reachability properties. We speculate that the reason for the lack of support is less in the technical challenges they pose, but the lack of a use case. Processes are expected to specify how a `good state' can be reached. 

The second limitation is the size of the model. We are confident that a holistic analysis of multiple protocols acting in parallel can be conducted for the whole of \citeauthor{speicher2018formally}'s model, but what if we want to include a full-grown model of different versions of TLS and IPsec? A deeper exploration of protocol composition in light of the infrastructure attacker and IP-like communication may yield a refined verification methodology and perhaps even composition results for IA models.

\AtNextBibliography{\small}
\printbibliography

\appendix
\subsection{ProVerif queries}
\label{appendix}
\lstset{
  language={Proverif},
  basicstyle=\scriptsize\ttfamily, 
  columns=fixed, 
  keepspaces=true, 
  showstringspaces=false, 
  breaklines=true, 
  numbers=none
}
\begin{full}
\begin{figure}
\centering
\begin{lstlisting}[label=queries,breaklines=true,caption=Queries in ProVerif]
query m:provider, n:provider, m':dom, n':dom, 
e:ip, d:dom, g:ip, r:ip, i:ip, j:ip;
event(Unconf(m,n))  
    	==>   (event(isMailserver(m',m))
               && event(A_record(i,m'))
               && event(C_ip(i)))
          ||  (event(isMailserver(n',n))
               && event(A_record(i,n'))
               && event(C_ip(i)))
          ||  (event(isMailserver(m',m))
               && event(A_record(i,m')) 
               && event(Received(n,d,r))
               &&(
                  (event(queries_prov(i,n))
                  && event(Resolver(i,g))
                  && event(C_ip(g)))
               ||(event(queries_prov(i,n))
                  && event(Resolver(i,g))
                  && event(UsedDomServer(g,e))
                  && event(C_ip(e)))
               ||(event(queries_prov(i,n))
                  && event(Resolver(i,g))
                  && event(C_routing(i,g)))
               ||(event(queries_prov(i,n))
                  && event(Resolver(i,g))
                  && event(UsedDomServer(g,e))
                  && event(C_routing(g,e))
                  && event(nDNSSEC(n))))
                  )
          ||  (event(isMailserver(m',m))
               && event(A_record(i,m'))
               && event(queries_prov(i,n))
               && event(Received(n,d,j))
               && event(C_routing(i,j))). 


query x:provider, d:dom, m:ip, e:ip, f:ip, 
g:ip;
event(Received(x,d,m)) 
    	  ==> (event(Register_MX(x,d)) 
               && event(Register_A(d,m)))
         ||   (event(queries_prov(f,x))
               && event(C_ip(f)))
         ||   (event(queries_prov(f,x))
               && event(Resolver(f,g))
               && event(C_ip(g)))
         ||   (event(queries_prov(f,x))
               && event(Resolver(f,g))
               && event(C_routing(f,g)))
         ||   (event(queries_prov(f,x))
               && event(Resolver(f,g))
               && event(UsedDomServer(g,e))
               && event(C_routing(g,e))
               && event(nDNSSEC(x)))
         ||   (event(queries_prov(f,x))
               && event(Resolver(f,g))
               && event(UsedDomServer(g,e))
               && event(C_ip(e))).
\end{lstlisting}
\end{figure}

\end{full}
\section{Background: Complete threat model}
\label{app:threat}

In this section, we will present the complete threat model described in
\cite{speicher2018formally}. Hence, the following will be
freely, but completely cited from \cite{speicher2018formally}, except
for the modifications marked in \newstuff{bold orange}.

We will give each rule, followed by the intuition of what kind of
attack it represents.

\subsubsection{Initially Compromised Nodes}

\begin{myrule}[ $r_\mathit{init-loc}$ in \cite{speicher2018formally} ]
All autonomous systems, IPs and domains associated to the attacking country
are initially under control of the attacker.
\[
\infer{\compr(x)}{x \in \AS \cup \addr \cup \domain & cn \in \countr & x \xrightarrow{\loc} n & \compr(cn)}
\]
\end{myrule}

\begin{myrule}[$r_\mathit{init-as}$\cite{speicher2018formally}]
If an AS is under the control of the attacker,
any IP which is part of the AS is also under control of the attacker.
\[
\infer{\compr(i)}{i \in \addr ~ a \in \AS ~ i \xrightarrow{\orig} a ~ \compr(a)}
\]
\end{myrule}

\begin{myrule}[$r_\mathit{init-dom}$\cite{speicher2018formally}]
If an IP is under the control of the attacker,
any domain that resolves to it (even if the attacker cannot interfere
with the resolution) is also under the control of the attacker.
\[
\infer{\compr(d)}{d \in \domain ~ i \in \addr ~ d \xrightarrow{\A} i ~ \compr(i)}
\]
\end{myrule}

\begin{myrule}
    [$r_\mathit{init-ip}$ \newstuff{(this rule is new)}]
If a domain is under the control of the attacker,
any IP it resolves to (even if the attacker cannot interfere
with the resolution) is also under the control of the attacker.
\[
\infer{\compr(i)}{d \in \domain ~ i \in \addr ~ d \xrightarrow{\A} i ~ \compr(d)}
\]
\end{myrule}

\subsubsection{Attacks via Routing}

\begin{myrule}
[ $r_\mathit{injection}$\cite{speicher2018formally} ]
If the attacker controls an AS which transfers packets from a domain $m$
to some IP address belonging to $n$
and this particular connection is not secured via the VPN mitigations,
we assume that the integrity of the communication from $m \in \domain$
to $n \in \domain$ is compromised.
\[
\infer{\intrrd(\newstuff{i, j})}{\substack{%
\rem{$d, e \in \domain$} \quad i, j \in \addr \quad a, b, c \in \AS
\quad \compr(b) \quad \vpn(a,c)
\\
\rem{$d \xrightarrow{\A}$} i \quad \rem{$e \xrightarrow{\A} j$} \quad i \xrightarrow{\orig}
a \quad j \xrightarrow{\orig} c \quad
a \xrightarrow{\RTE(b)} c
\\
}}
\]
\end{myrule}

On the resolution and application level we are only concerned with
communication between domains. Thus this rule covers all
relevant routing attacks.

\subsubsection{Integrity of domain/MX resolution}

\begin{myrule}
[ $r_\mathit{dns-ns}$\cite{speicher2018formally} ]
If the attacker controls any name server that could be queried during
resolution,
we consider the integrity of the domain name resolution compromised.
\[
\infer{\intrh(d)}{d, e \in \domain & d \xrightarrow{\DNS} e & \newstuff{$e \xrightarrow{\A} i$} & \compr(\newstuff{$i$})}
\]
\begin{full}
Although unspecified by RFC~3207,
name servers commonly attach  an \A record whenever they respond to
a resolution request with an \DNS entry pointing to another name
server~\cite{rfc3207}. This speeds up the resolution and has the pleasant
side-effect that the integrity of the resolution does not depend on
whether these authoritative
the integrity of the resolution of the domain names of the name servers requested,
hence the transitive rule for $\intrh$ is not at the attacker's
disposal.
\end{full}
\end{myrule}

\begin{myrule}
[ $r_\mathit{dns-res}$\cite{speicher2018formally} ]
If the attacker controls the resolver of a domain,
we consider the integrity of any domain name resolution this domain
attempts compromised.
(Technically, $r$ is an IP address, but we simplified this and the following
rule for presentation.)
\[
\infer{\intrd(d, e)}{d, e \in \domain & i \in \addr & d \xrightarrow{\RES}
i & \compr(i)}
\]
\end{myrule}

\begin{myrule}
[$r_\mathit{dns-route-res}$\cite{speicher2018formally} ]\label{dnsrouteres}
If the attacker controls the route from a domain to the resolver this
domain uses,
we consider the integrity of any domain name resolution this domain
attempts compromised \rem{,
unless the integrity of the resolution is guaranteed by DNSSEC}.
\[
\infer{\intrd(d, e)}{d, e  \in \domain & \newstuff{$i \in \addr$}
& d \xrightarrow{\RES} r & \newstuff{$d \xrightarrow{\A} i$}
& \intrrd(\newstuff{$i$}, r) & \rem{\dnssec(e)}}
\]
\end{myrule}

\begin{myrule}
[ $r_\mathit{dns-route-ns}$\cite{speicher2018formally} ]
If the attacker controls the route from a resolver to some authoritative name server
potentially queried during resolution,
we consider the integrity of the resolution for this domain name
compromised,
unless the integrity of the resolution is guaranteed by DNSSEC.
\[
\infer{\intrd(d, e)}{\substack{d, e, f  \in \domain \quad \newstuff{$r \in \addr$} \quad d \xrightarrow{\RES} r
\quad e \xrightarrow{\DNS} f \\ \newstuff{$f \xrightarrow{\A} i$} \quad \intrrd(r, \newstuff{i}) \quad \dnssec(e)}}
\]
\end{myrule}

\subsubsection{Confidentiality}

\begin{myrule}
[ $r_\mathit{compromise}$\cite{speicher2018formally} ]
If a mail server is already compromised, e.g.,
if it is hosted by an adversarial country,
the attacker can compromise the confidentiality of the
communication between two mail providers.

\[
\infer{\unconf(d, e)}{\substack{
d, e \in \provider
\quad d \xrightarrow{\MX} d' \quad e \xrightarrow{\MX} e' \\
\newstuff{$d' \xrightarrow{\A} d''$} \quad \newstuff{$e' \xrightarrow{\A} e''$}
\quad \compr(\newstuff{$e'$}') \lor
\compr(\newstuff{$d''$})}
}
\]
\end{myrule}

\begin{myrule}
[ $r_\mathit{fake-mx}$\cite{speicher2018formally} ]
If the sender does not enforce strict host validation,
e.g., by using optimistic STARTTLS,
the attacker can compromise the confidentiality of the
communication between two mail providers
by changing a provider's MX record to point to a domain under her
control.

\[
    \infer{\unconf(d, e)}{\substack{
d, e \in \provider \quad
d\neq e \quad
d \xrightarrow{\MX} d'
\\
\tls^\sender(d) \quad
\intrh(e) \lor \intrd(d',e)
}}
\]
\end{myrule}

\begin{myrule}
[ $r_\mathit{fake-ip}$\cite{speicher2018formally} ]
If the sender does not enforce strict host validation,
\begin{full}
e.g., by using
optimistic STARTTLS,
\end{full}
the attacker can compromise the confidentiality of the
communication between two mail providers
by pointing the domain of the MX to an
IP of her choice.
\[
\infer{\unconf(d, e)}{\substack{
d, e \in \provider \quad
d\neq e \quad
d \xrightarrow{\MX} d' \quad e \xrightarrow{\MX} e' \\
\intrh(e') \lor \intrd(d', e') \quad \tls^\sender(d)}}
\]
\end{myrule}

\begin{myrule}
[ $r_\mathit{intercept}$\cite{speicher2018formally} ]
If the sender does not enforce strict host validation,
\begin{full}
e.g., she is using optimistic STARTTLS,
and DANE is not deployed,
\end{full}
the attacker can compromise the confidentiality of the
communication between two mail providers
by
intercepting packets on the
route between their respective mail servers.
\[
\infer{\unconf(d, e)}{\substack{
d, e \in \provider \quad
d\neq e \quad
d \xrightarrow{\MX} d' \quad e \xrightarrow{\MX} e' \\
\newstuff{$d' \xrightarrow{\A} d''$} \quad \newstuff{$e' \xrightarrow{\A} e''$} \quad
\intrrd(\newstuff{$d'', e''$}) \quad \tls^\sender(d)
\quad \dane^\receiver(e)
}}
\]
\end{myrule}

\begin{myrule}
[ $r_\mathit{fake-mx-strict}$\cite{speicher2018formally} ]
If the sender does not enforce certificate validation according to
RFC 7817, e.g., by using
optimistic STARTTLS or strict validation on the hostname only,
the attacker can compromise the confidentiality of the
communication between two mail providers
by changing a provider's MX record to point to a domain under her
control.
\[
\infer{\unconf(d, e)}{\substack{
d, e \in \provider \quad
d\neq e \quad
d \xrightarrow{\MX} d' \\
\intrh(e) \lor \intrd(d', e) \quad \strictValidation(d)}}
\]
\end{myrule}

\subsection{Lemmas}
\label{appendlemma}

Before proving Theorem~\ref{the:main}, we present the proofs to the 
two
process transformations from Section~\ref{justify} .

\begin{proof}[Proof of Lemma \ref{lem:in}]
We can show that for every trace of the process $in(v).Q'$, the same trace can be produced by $Q$. In the first process, the adversary has to choose what term it binds to the free variable $v$ in the beginning. Therefore, it needs to deduce $T$ from the frame, s.t. $\frameabr \vdash t$. Comparing to $Q$, we can deduct the same term $T$ by replacing $in(v).Q'$ with $Q$ in the same configuration. Since no rule of our operational semantics deletes any information from the frame, we are able to deduce $T$ at any point during the process execution of $Q$. This allows us to substitute $v$ with $T$ in both processes, leading to the same set of traces (since the rest of the processes is the same by construction.)
\end{proof}

%
%

\begin{proof}
To proof Lemma \ref{con:replin}, we start by applying (\textup{repl}) (see Figure \ref{semantics}) $\abs{\overrightarrow{\textit{PN}}}$ times
to the process $!in(\vecv).P)$ and get a new process  $Q = !in(\vecv).P) \underbrace{\mid in(\vecv).P) \mid ...\mid in(\vecv).P)}_{\abs{\overrightarrow{\textit{PN}}}}$.
The variables $\vecv$ in the right process of Lemma \ref{con:replin} are substituted by public names provided by $G$. We apply the rule (\textup{in}) also $\abs{\overrightarrow{\textit{PN}}}$  times on $Q$ and the adversary can input the same public names as provided by $G$ since all public names are deducible from the frame. With this transformation of $Q$ we get exactly $!in(\vecv).P) \mid \bigPar\limits_{\vecin{p}}P\set{\substi{\vecin{p}}{\vecv}}$.
Hence, we can conclude that
\begin{align*}
traces(!in(\vecv).P)\ &= traces(!in(\vecv).P) \mid \bigPar\limits_{\vecin{p}}P\set{\substi{\vecin{p}}{\vecv}}) \\
&\supseteq\ traces(\bigPar\limits_{\vecin{p}}P\set{\substi{\vecin{p}}{\vecv}})
\end{align*}
\end{proof}

\begin{proof}[Proof of Theorem \ref{the:main}]
    \newcommand{\Pproto}{P_\mathname{proto}}
First, we rearrange the $||$-quantification 
in  $\mathcal{F}(G)$ 
(see Figure~\ref{paraprocess}),
so that
$\vecin{p}$ consists of all assignments to the
meta-variables\footnote{%
The $||$-notation is a syntactic shortcut on the mathematical
level, hence the variables it binds are mathematical variables,
not ProVerif variables. They 
stand for nodes in the graph, which we assumed to be public names.}
$x$, $y$ and $z$%
.
(By definition, $||$ is associative and commutative.)
For the reader's convenience, we index the applied-$\pi$ variables
with
the meta-variable they replace.
\begin{align*}
\traces(\mathcal{F}(G)) & = \traces
\left(
    \bigPar\limits_{\vecin{p}} \Pproto \substi{\vecin{p}}{\vecin{v}}
\right),
\intertext{where 
$\Pproto=
    !P_\smtpclient(\vecin{x}^\client,\vecin{x}^\server,\vecin{x}^\RES)
\mid
!P_\smtpserver\allowbreak(\vecin{y}^\server,\vecin{y}^\client)
\mid
!P_\dnsres(\vecin{z}^\RES, \vecin{z}^\client, \vecin{z}^\DNS, \vecin{z}^\rns)
\mid
!P_\dnsns(\vecin{u}^\DNS,\allowbreak \vecin{u}^\RES)
\mid
!P_\dnsrns(\vecin{w}^\rns, \vecin{w}^\RES)
) 
$ (compare with Figure~\ref{paraprocess}).
Therefore, by applying Lemma~\ref{con:replin} we get:
}
  & \subseteq \traces\left(
      !\pin(\vecv).\Pproto
      \right)
\intertext{
Note that all variables are uniquely named.
We can hence exhaustively apply Lemma~\ref{lem:in}
to $P$ to push all variables in $\vecv$ to the inside far enough that 
the resulting process  
      $\Pproto'$
      equals $P$ (compare with full model in 
      \appendixorfull{sec:full-model}.). 
      Verifying the syntactical equivalence, we obtain:
}
  & \subseteq \traces(P).
\end{align*}
\end{proof}

\begin{full}
\subsection{Full model}\label{sec:full-model}
\lstset{
  language={Proverif},
  basicstyle=\small\ttfamily, 
  columns=fixed, 
  keepspaces=true, 
  showstringspaces=false, 
  breaklines=true, 
  commentstyle= \color{lightgray},
  morecomment = [l]{//},
  morecomment = [n][\color{orange}]{(*}{*)},
}
\onecolumn
\begin{lstlisting}[label=Fullmodel,caption=Full ProVerif model]
free c: channel.

(* 

This model is  a process used to establish the soundness results in 

Dax, Kunnemann: On the Soundness of Infrastructure Adversaries

It covers network of clients and servers running the SMTP, DNS and DNSSEC
protocols in a setting that civers a simplified version of inter-AS
communication. Hence the attacker must corruption communication channels before
being able to eavesdrop.


Notes:

We use 9 private constructor symbols, 6 for modeling communication that is
secure until corruption, and 4 to deal with state that is secure until
corruption.

**Communication:** The protocol for establishing a communication channel is as
follows. Say A wants to communicate with B.

1. At initialization time, A broadcasts information about himself: IP, domain,
   AS, depending on the context, association of an abstract provider (e.g., "I
   am a mail server of google"). This information is bundled in transparent
   constructors, i.e., constructors that can be deconstructed without any
   additional secrets. Their purpose is to keep the model maintainable, and
   they are marked with a `_info` suffix.  Their type is `service`.
2. (same for B)
3. A is instructed by the adversary whom to communicate with --- this is the
   standard way of modeling arbitrary communication patterns. The communication
   partner is identified by a term of type `service` as output in step 1.$^1$
4. (same for B)
5. A choses a source port and uses `c_communicate_src_port` to communicate it
   to B. The function symbol is private, but it has a destructor for the source
   port, meaning that this information is revealed to the attacker. It does not
   have a destructor for the description of the target service, which is
   admittedly unintuitive, but was necessary to speed up verification time. As
   this information is shared at initialization time, it is always available to
   the attacker anyway. We typed this parameter (type `service`), so we can
   ensure that the function symbol is indeed used that way.
6. B unpacks this source port using the destructor for
   `c_communicate_src_port`.
7. Then the source ports are exchanged in the other direction, B repeats step
   5, but for the function symbol `r_communicate_src_port`.
8. A receives the source port with the destructor for `r_communicate_src_port`.
9. Now both parties can build the key that identifies these channels with the
   function symbol `chanbuilder`. There are, in fact, two channels, one for
   messages from A to B, and another for messages from B to A. A party that
   knows the channel (i.e., a `chanbuilder-`term, i.e., a term of type `chan`)
   can obtain the message on this channel using the destructor `get_req_packet`.

$^1$ The exception is the SMTP client process, in which the attacker supplies
only the provider (e.g., "google.com"), as this is the information that the
SMTP client has in the real-world (e.g., through the To: field in the email).
Instead, the IP and domain are obtained using the DNS resolver process.

There are five private function symbols for communicating the source port that
way, one for every role in the protocol, plus the `chanbuilder` function symbol
that sets up the channel (of type `chan`). Why is such a complicated protocol
necessary? As outlined in the paper, we use the `chanbuilder` terms to
establish private-but-asynchronous channels. Hence the need for `chanbuilder`
to be private. As in the IP protocol, the channel is identified by source and
destination IP and port. Hence, these need to be communicated to set up that
channel. In the IP protocol, this is part of a handshake. An alternative and
seemingly more natural modeling would be to allow to extract the source port
from a `chanbuilder`-term, however, in that modeling, a channel could only be
corrupted once the first message has been sent. Our modeling ensures that
a channel can be corrupted before the first message is being sent. This
explains the need to send the source port before the first message. The message
that communicates the source port is in a private constructor to ensure
authenticity, so that a channel cannot be manipulated *before* it is corrupted.
The terms within that constructor are public or can be computed. We have five
different versions of this constructor, one for each party, to again improve
verification time. This is a valid modeling choice because in the real world,
the message that communicates the source port identifies the intended protocol
type and role of the sender via the destination port.


**Storage**: The four remaining private function symbols are used instead of
a global storage (we use ProVerif's `table` feature only for the publicly known
public keys of DNSSEC root servers). They represent local database entries
created on initiasation.
        - `register_server` binds an SMPT servers domain to an IP. This is only
          used by the authoritative nameserver, where it represents its local
          registration DB. They are generated in the top-level process upon
          initalisation, where they are controlled by the attacker, but cannot
          be changed later. It is used to answer `A`-requests.
        - `register_provider` binds providers (the part of email addresses
          after the "@") to the domains providing SMTP services for them. It is
          set up just like `register_server` but is used to handle
          `MX`-requests.
   - `register_ns` binds nameserver to their public keys. `register_ns` terms
     are used by the root server to provide the public key of the authoritative
     nameserver directly below it. They, again, correspond to a local
     registration database, are used nowhere else and are controlled by the
     attacker upon initialisation, but not later.
   - `valid_prov` marks a valid provider, e.g., "google.com". This corresponds
     to the public (or the part of the public covered when Speicher et al
     consider, e.g., the Top-10 providers) knowing which provider they use.

There are at least two other established ways of implementing a local store.
First are ProVerif's tables, which implements a key/value store. Unfortunately,
it was not readily clear what could provide a key. Neither domain, IP or AS are
unique identifiers, hence values would risk being overwritten. The second are
private channels, which, giving our experiences with the modelling of
communication channels, we were worried about slowing down verification.

*)



(* types *)
type provider.
type dom.
type ip.
type as.
type port.
type service.
type chan.
type regis.
type com.

(* DNSSEC PKI trustbase*)
table trustbase(port, bitstring).

(* global names*)
free CLIENT_PORT:port.
free SERVER_PORT:port.
free RES_PORT:port.
free NS_PORT:port.
free ROOT_PORT:port.

free OK:bitstring.

(* signature *)
fun vk(bitstring):bitstring. (* vk/1 *)
fun sign(bitstring,bitstring):bitstring. (* sign/2 *)
reduc forall sk:bitstring, m:bitstring; readsign(sign(sk,m)) = m.
reduc forall sk:bitstring, m:bitstring; verify(vk(sk), m, sign(sk,m)) = true.

(* pairs *)
fun pair(bitstring,bitstring):bitstring [data].


(* mx pairs *)
fun mx_pair(provider,dom):bitstring.
reduc forall x:provider, y:dom; mx_fst(mx_pair(x,y))=x.
reduc forall x:provider, y:dom; mx_snd(mx_pair(x,y))=y.

(* rr pairs *)
fun rr_pair(dom,ip):bitstring.
reduc forall x:dom, y:ip; rr_fst(rr_pair(x,y))=x.
reduc forall x:dom, y:ip; rr_snd(rr_pair(x,y))=y.

(* ds pairs *)
fun ds_pair(service,bitstring):bitstring.
reduc forall x:service, y:bitstring; ds_fst(ds_pair(x,y))=x.
reduc forall x:service, y:bitstring; ds_snd(ds_pair(x,y))=y.


(* dnssec response wrapper *)
fun dnssec_resp(bitstring, bitstring, bitstring, bitstring, bitstring):bitstring [data].

(* request A record*)
fun request_dom(dom):bitstring.
reduc forall x:dom; get_dom_request(request_dom(x))=x.

(* request MX record*)
fun request(provider):bitstring.
reduc forall x:provider; get_request_prov(request(x))=x.

(* dnssec request A record *)
fun dnssec_request(dom):bitstring.
reduc forall x:dom; get_dnssec_dom(dnssec_request(x))=x.

(* dnssec request MX record *)
fun dnssec_request_prov(provider):bitstring.
reduc forall x:provider; get_dnssec_prov(dnssec_request_prov(x))=x.

(* answer A record *)
fun answer(bitstring,ip,dom):bitstring.
reduc forall x:bitstring, y:ip, z:dom; check_answer_req(answer(x,y,z))=x.
reduc forall x:bitstring, y:ip, z:dom; get_answer_ip(answer(x,y,z))=y.
reduc forall x:bitstring, y:ip, z:dom; get_answer_dom(answer(x,y,z))=z.

(* answer MX record *)
fun answer_dom(bitstring,provider,dom):bitstring.
reduc forall x:bitstring, y:provider, z:dom; check_req_answer(answer_dom(x,y,z))=x.
reduc forall x:bitstring, y:provider, z:dom; get_prov_answer(answer_dom(x,y,z))=y.
reduc forall x:bitstring, y:provider, z:dom; get_dom_answer(answer_dom(x,y,z))=z.

(* SMTP message wrapper *)
fun smtp_req(ip,ip,bitstring):bitstring.
reduc forall x:ip, y:ip, z:bitstring; smtp_fst(smtp_req(x,y,z))=x.
reduc forall x:ip, y:ip, z:bitstring; smtp_snd(smtp_req(x,y,z))=y.
reduc forall x:ip, y:ip, z:bitstring; smtp_third(smtp_req(x,y,z))=z.

(* private channel wrapper using secret value x:chan - requests *)
fun req_packet(chan,bitstring):bitstring.
reduc forall x:chan, y:bitstring; get_req_packet(x,req_packet(x,y))=y.

(* private channel wrapper using secret value x:chan - answers *)
fun ans_packet(chan,bitstring):bitstring.
reduc forall x:chan, y:bitstring; get_ans_packet(x,ans_packet(x,y))=y.

(* public client information *)
fun client_info(provider,as,ip,dom):service.
reduc forall a:provider, b:as, x:ip, d:dom; get_client_prov(client_info(a,b,x,d))=a.
reduc forall a:provider, b:as, x:ip, d:dom; get_client_as(client_info(a,b,x,d))=b.
reduc forall a:provider, b:as, x:ip, d:dom; get_client_ip(client_info(a,b,x,d))=x.
reduc forall a:provider, b:as, x:ip, d:dom; get_client_dom(client_info(a,b,x,d))=d.

(* public server information *)
fun server_info(provider,as,ip,dom):service.
reduc forall a:provider, b:as, x:ip, d:dom; get_server_prov(server_info(a,b,x,d))=a.
reduc forall a:provider, b:as, x:ip, d:dom; get_server_as(server_info(a,b,x,d))=b.
reduc forall a:provider, b:as, x:ip, d:dom; get_server_ip(server_info(a,b,x,d))=x.
reduc forall a:provider, b:as, x:ip, d:dom; get_server_dom(server_info(a,b,x,d))=d.

(* public resolver information *)
fun res_info(dom,ip,as):service.
reduc forall a:dom, b:ip, x:as; get_res_dom(res_info(a,b,x))=a.
reduc forall a:dom, b:ip, x:as; get_res_ip(res_info(a,b,x))=b.
reduc forall a:dom, b:ip, x:as; get_res_as(res_info(a,b,x))=x.

(* public NS information *)
fun dns_info(dom, ip, as):service.
reduc forall a:dom, b:ip, x:as; get_dns_dom(dns_info(a,b,x))=a.
reduc forall a:dom, b:ip, x:as; get_dns_ip(dns_info(a,b,x))=b.
reduc forall a:dom, b:ip, x:as; get_dns_as(dns_info(a,b,x))=x.

(* public root NS information *)
fun root_info(dom, ip, as):service.
reduc forall a:dom, b:ip, x:as; get_root_dom(root_info(a,b,x))=a.
reduc forall a:dom, b:ip, x:as; get_root_ip(root_info(a,b,x))=b.
reduc forall a:dom, b:ip, x:as; get_root_as(root_info(a,b,x))=x.

(* wrapper for NS resolution, root NS requests*)
fun ask(service):bitstring.
reduc forall x:service; get_ask(ask(x))=x.

(* The following constructor symbols are flagged as private. All these constructors are used to 
exchange information between subprocesses. As there are no pre-shared secrets and we require 
authenticated information exchange, we use these private constructors as trusted channels. 
With flagging constructors as private, the adversary cannot use the constructor and hence all 
occurrences of these constructors are honestly generated. The authenticated information exchange
is needed by our custom threat model as we require an adversary to corrupt a connection before 
it can use its Dolev-Yao model power. *)

(* trusted registration for A records of mailservers *)
fun register_server(dom,ip):regis [private].
reduc forall x:dom, y:ip; get_ip(x,register_server(x,y))=y.

(* register MX records *)
fun register_provider(provider,dom):service [private].
reduc forall a:provider, x:dom; get_valid_prov(register_provider(a,x))=a.
reduc forall a:provider, x:dom; get_valid_dom(register_provider(a,x))=x.

(* providers who registered a mailserver able to receive messages *)
fun valid_prov(provider):provider [private].
reduc forall a:provider; get_prov_valid(valid_prov(a))=a.

(* trusted registration for name servers *)
fun register_ns(service,bitstring):regis [private].
reduc forall x:service, y:bitstring; get_chan(x,register_ns(x,y))=y.

fun c_communicate_src_port(service,port):com [private].
reduc forall x:service, y:port; c_get_src_port(x,c_communicate_src_port(x,y))=y.

fun s_communicate_src_port(ip,port):com [private].
reduc forall x:ip, y:port; s_get_src_port(x,s_communicate_src_port(x,y))=y.

fun r_communicate_src_port(service,port):com [private].
reduc forall x:service, y:port; r_get_src_port(x,r_communicate_src_port(x,y))=y.

fun d_communicate_src_port(service,port):com [private].
reduc forall x:service, y:port; d_get_src_port(x,d_communicate_src_port(x,y))=y.

fun a_communicate_src_port(service,port):com [private].
reduc forall x:service, y:port; a_get_src_port(x,a_communicate_src_port(x,y))=y.




(* private constructor to build channels between two parties*)
fun chanbuilder(ip,port,ip,port):chan [private].




(* events *)
event Received(provider,dom,ip).
event Register_A(dom,ip).
event Register_MX(provider,dom).
event C_ip(ip).
event C_as(as).
event queries(ip,dom).
event queries_prov(ip,provider).
event Resolver(ip,ip).
event C_routing(ip,ip).
event UsedDomainServer(ip,ip).
event isMailserver(dom,provider).
event Routing(as,as,as).
event Unconf(provider,provider).
event A_record(ip,dom).
event nDNSSEC(provider).
event IPinAS(ip,as).






(* Unconf Query *)
query m:provider, n:provider, m':dom, n':dom, e:ip, d:dom,
      f:dom, g:ip ,x:as, y:as, z:as, r:ip, i:ip, j:ip;
    event(Unconf(m,n))  ==>   (event(isMailserver(m',m)) (*r-compromise*)
                            && event(A_record(i,m'))
                            && event(C_ip(i)))
                        ||  (event(isMailserver(n',n))  (*r-compromise*)
                            && event(A_record(i,n'))
                            && event(C_ip(i)))
                        ||  (event(isMailserver(m',m))  (*r-fake-mx,*)
                            && event(A_record(i,m'))    (*r-fake-ip,*)
                            && event(Received(n,d,r))   (*r-fake-mx-strict*)
                            &&(
                              (event(queries_prov(i,n))
                              && event(Resolver(i,g))
                              && event(C_ip(g)))
                            ||(event(queries_prov(i,n))
                              && event(Resolver(i,g))
                              && event(UsedDomainServer(g,e))
                              && event(C_ip(e)))
                            ||(event(queries_prov(i,n))
                              && event(Resolver(i,g))
                              && event(C_routing(i,g)))
                            ||(event(queries_prov(i,n))
                              && event(Resolver(i,g))
                              && event(UsedDomainServer(g,e))
                              && event(C_routing(g,e))
                              && event(nDNSSEC(n)))))
                        ||  (event(isMailserver(m',m))  (*r-intercept*)
                            && event(A_record(i,m'))
                            && event(queries_prov(i,n))
                            && event(Received(n,d,j))
                            && event(C_routing(i,j))). 


(* Integrity Query *)
query x:provider, a:as, b:as, z:as, d:dom, n:ip, m:ip, p:ip, e:ip, f:ip, g:ip;
    event(Received(x,d,m)) ==> (event(Register_MX(x,d)) 
                          && event(Register_A(d,m))) (*Sanity check*)
                        ||   (event(queries_prov(f,x)) (*Sanity check*)
                          && event(C_ip(f)))
                        ||  (event(queries_prov(f,x))(*r-dns-res*)
                          && event(Resolver(f,g))
                          && event(C_ip(g)))
                        ||  (event(queries_prov(f,x)) (*r-dns-route-res*)
                          && event(Resolver(f,g))
                          && event(C_routing(f,g)))
                        ||  (event(queries_prov(f,x)) (*r-dns-route-ns*)
                          && event(Resolver(f,g))
                          && event(UsedDomainServer(g,e))
                          && event(C_routing(g,e))
                          && event(nDNSSEC(x)))
                        ||  (event(queries_prov(f,x)) (*r-dns-ns*)
                          && event(Resolver(f,g))
                          && event(UsedDomainServer(g,e))
                          && event(C_ip(e))).


(* Root Nameserver *)
let root_server(dom_root:dom, ip_root:ip, AS:as) =
    (* Receive public information of the resolver process *)
    in(c, reswrap:service);
    !(
    (* Choose fresh source port for each session *)
    new port_root:port;
    out(c,port_root);
    (* Send src port over trusted channel to communication partner. The previous line seems redundant since the 
    adversary can gain the same knowledge as the key to the private constructor is publically known. We keep this redundance because
    the first line expresses that the port is publically known, whereas the second represents that it is sent over a trusted channel. *)
    out(c, a_communicate_src_port(root_info(dom_root,ip_root,AS),port_root));
    (* Choose new DNSSEC keys each session.*)
    new ksk:bitstring;
    out(c,vk(ksk));
    new zsk:bitstring;
    out(c,vk(zsk));
    (* distribute the public part of the key signing key to the trustbase *)
    insert trustbase(port_root, vk(ksk));
    let ip_res = get_res_ip(reswrap) in
    (* Build private channel over IPs, the fresh src port and the *)
    (* public target port. *)
    let rootchan = chanbuilder(ip_root,port_root,ip_res,RES_PORT) in
    in(c, res_packed_src_port:com);
    let port_res = r_get_src_port(reswrap,res_packed_src_port) in
    let rootchan2 = chanbuilder(ip_res,port_res,ip_root,ROOT_PORT) in
    (
    ( (* Receive a request from a resolver *)
      in(c, req:bitstring);
        let pack = get_req_packet(rootchan2, req) in
        (* Let adversary provide the information of the *)
        (* nameserver in charge of the requested zone *)
        in(c ,dnswrap:service);
        let askdns = ask(dnswrap) in
          out(c, req_packet(rootchan,askdns))   
          )
    |
    ((* Receive a request from a resolver *)
      in(c, dnssec_req:bitstring);
      let dnssec_pack = get_req_packet(rootchan2, dnssec_req) in
        (* Let adversary provide the information of the *)
        (* nameserver in charge of the requested zone *)
        in(c ,dnssecwrap:service);
          (* Check that the nameserver is registered for *)
          (* DNSSEC and receive its public ksk *)
          in(c, registration_ns:regis);
          let pk_dns = get_chan(dnssecwrap,registration_ns) in
            (* Prepare DS record and answer the resolver *)
            let answer_dns = ds_pair(dnssecwrap,pk_dns) in
              let rr = pair(vk(zsk),vk(ksk)) in
                let rr_sig = sign(ksk, rr) in
                  let ds = answer_dns in
                    let ds_sig = sign(zsk, ds) in
                      let dnssec_comb_message = dnssec_resp(answer_dns,rr,rr_sig,ds,ds_sig) in
                        out(c, req_packet(rootchan,dnssec_comb_message))   
          )
    ))
    .

(* Authorative Nameserver *)
let dns_server(dom_dns:dom,ip_dns:ip,AS:as) =
  (* Receive public information of the resolver process *)
    in(c, reswrap:service);
    !( 
    (* Choose fresh source port for each session *) 
    new port_dns:port;
    (* Give source port to the adversary *)
    out(c,port_dns);
    (* Send src port over trusted channel to communication partner. The previous line seems redundant since the 
    adversary can gain the same knowledge as the key to the private constructor is publically known. We keep this redundance because
    the first line expresses that the port is publically known, whereas the second represents that it is sent over a trusted channel. *)
    out(c,d_communicate_src_port(dns_info(dom_dns,ip_dns,AS),port_dns));
    (* Choose new DNSSEC keys each session.*)
    new ksk:bitstring;
    out(c,vk(ksk));
    new zsk:bitstring;
    out(c,vk(zsk));
    (* distribute the public part of the key signing key to the trustbase *)
    out(c, register_ns(dns_info(dom_dns,ip_dns,AS), vk(ksk)));
    (* Build private channel over IPs, the fresh src port and the *)
    (* public target port. *)
    let ip_res = get_res_ip(reswrap) in
    let dnschan = chanbuilder(ip_dns,port_dns,ip_res,RES_PORT) in
    in(c, port_packet_src_res:com);
    let port_res = r_get_src_port(reswrap,port_packet_src_res) in
    let dnschan2 = chanbuilder(ip_res, port_res, ip_dns, NS_PORT) in
    (
    ( (* Corrupted IP -> Keys and Channels are leaked to the adversary *)
    event C_ip(ip_dns);
    out(c,ksk);
    out(c,zsk);
    out(c,dnschan);
    out(c,dnschan2))
    |
    ((* Corrupted AS -> all IPs located in this AS are leaked -> *)
    (* Corrupted IP -> Keys and Channels are leaked to the adversary *)
    event C_as(AS);
    event C_ip(ip_dns);
    out(c,ksk);
    out(c,zsk);
    out(c,dnschan);
    out(c,dnschan2))
    |
    ((* Adversary chooses an AS to be corrupted and chooses to route *)
    (* over this AS. This leaks the channel over this route to be leaked *)
    in(c,inter2:as);
    let as_res = get_res_as(reswrap) in
    event Routing(as_res, inter2, AS);
    event C_as(inter2);
    event C_routing(get_res_ip(reswrap),ip_dns);
    out(c,dnschan))
    |
    ((* Adversary chooses an AS to be corrupted and chooses to route *)
    (* over this AS. This leaks the channel over this route to be leaked *)
    in(c,inter21:as);
    let as_res1 = get_res_as(reswrap) in
    event Routing(as_res1, inter21, AS);
    event C_as(inter21);
    event C_routing(get_res_ip(reswrap),ip_dns);
    out(c,dnschan2))
    |
    ( (* Receive MX record request from resolver*)
    in(c, req:bitstring);
    let real_pack = get_req_packet(dnschan2, req) in
          let provi = get_request_prov(real_pack) in
          (* Adversary provides a valid (priory registered) MX record *)
          in(c, valid:service);
            let provj = get_valid_prov(valid) in
            if provi = provj then
              let dom_prov = get_valid_dom(valid) in
                out(c, ans_packet(dnschan,answer_dom(real_pack, provi, dom_prov)))
            )
    |
    ( (* Receive A record request from resolver*)
    in(c, req2:bitstring);
    let real_pack2 = get_req_packet(dnschan2, req2) in
          let domi = get_dom_request(real_pack2) in
          in(c, registration:regis);
          (* Adversary provides a valid (priory registered) A record *)
              let ip4=get_ip(domi,registration) in
                out(c, ans_packet(dnschan,answer(real_pack2, ip4, domi)))
            )
    |
    ( (* Receive MX record, DNSSEC request from resolver*)
    in(c, dnssec_req:bitstring);
    let dnssec_pack = get_req_packet(dnschan2, dnssec_req) in
        let domi = get_dnssec_dom(dnssec_pack) in
        (* Adversary provides a valid (priory registered) A record *)
          in(c, registration:regis);
            let ip4=get_ip(domi,registration) in
            (* Prepare A record and answer the resolver *)
            let answer_dns = rr_pair(domi,ip4) in
              let rr = pair(vk(zsk),vk(ksk)) in
                let rr_sig = sign(ksk, rr) in
                  let a_rec = answer_dns in
                    let a_sig = sign(zsk, a_rec) in
                      let dnssec_comb_message = dnssec_resp(answer_dns,rr,rr_sig,a_rec,a_sig) in
                        out(c, ans_packet(dnschan,dnssec_comb_message))   
          )
    |
    ( (* Receive A record, DNSSEC request from resolver*)
    in(c, dnssec_req2:bitstring);
    let dnssec_pack = get_req_packet(dnschan2, dnssec_req2) in
        let provi = get_dnssec_prov(dnssec_pack) in
        (* Adversary provides a valid (priory registered) MX record *)
          in(c, valid:service);
          let provj = get_valid_prov(valid) in
            if provi = provj then
            let dom_prov = get_valid_dom(valid) in
            (* Prepare MX record and answer the resolver *)
            let answer_dns2 = mx_pair(provi,dom_prov) in
              let rr = pair(vk(zsk),vk(ksk)) in
                let rr_sig = sign(ksk, rr) in
                  let mx_rec = answer_dns2 in
                    let mx_sig = sign(zsk, mx_rec) in
                      let dnssec_comb_message2 = dnssec_resp(answer_dns2,rr,rr_sig,mx_rec,mx_sig) in
                        out(c, ans_packet(dnschan,dnssec_comb_message2))   
          )
    ))
    .

let resolver(dom_res:dom,ip_res:ip,AS:as) =
    (* Receive public information of the root NS process *)
    in(c, rootwrap:service);
    (* Receive public information of the client MX process *)
    in(c, clientwrap:service);
    !( 
    (* Choose fresh source port for each session *)  
    new port_res:port;
    out(c,port_res);
    (* Send src port over trusted channel to communication partner. The previous line seems redundant since the 
    adversary can gain the same knowledge as the key to the private constructor is publically known. We keep this redundance because
    the first line expresses that the port is publically known, whereas the second represents that it is sent over a trusted channel. *)
    out(c, r_communicate_src_port(res_info(dom_res,ip_res,AS),port_res));
    let root_ip = get_root_ip(rootwrap) in
    let client_ip = get_client_ip(clientwrap) in
    (* Build private channel over IPs, the fresh src port and the *)
    (* public target port. *)
    let resolverchan = chanbuilder(ip_res,port_res, root_ip, ROOT_PORT) in
    in(c, root_packed_src_port:com);
    let root_port = a_get_src_port(rootwrap,root_packed_src_port) in
    let resolverchan2 = chanbuilder(root_ip,root_port,ip_res,RES_PORT) in
    (* Build private channel over IPs, the fresh src port and the *)
    (* public target port. *)
    let resolverchan3 = chanbuilder(ip_res,port_res,client_ip,CLIENT_PORT) in
    in(c, client_packed_src_port:com);
    let client_port = c_get_src_port(clientwrap,client_packed_src_port) in
    let resolverchan4 = chanbuilder(client_ip, client_port,ip_res,RES_PORT) in
    event UsedDomainServer(ip_res,get_root_ip(rootwrap));
    (
    ((* Corrupted IP -> Keys and Channels are leaked to the adversary *)
    event C_ip(ip_res);
    out(c,resolverchan);
    out(c,resolverchan2);
    out(c,resolverchan3);
    out(c,resolverchan4))
    |
    ((* Corrupted AS -> all IPs located in this AS are leaked -> *)
    (* Corrupted IP -> Keys and Channels are leaked to the adversary *)
    event C_as(AS);
    event C_ip(ip_res);
    out(c,resolverchan);
    out(c,resolverchan2);
    out(c,resolverchan3);
    out(c,resolverchan4))
    |
    ((* Adversary chooses an AS to be corrupted and chooses to route *)
    (* over this AS. This leaks the channel over this route to be leaked *)
    in(c,inter1:as);
    let as_root = get_root_as(rootwrap) in
      event Routing(AS, inter1, as_root);
      event C_as(inter1);
      event C_routing(ip_res,get_root_ip(rootwrap));
      out(c,resolverchan))
    |
    (
    (* Adversary chooses an AS to be corrupted and chooses to route *)
    (* over this AS. This leaks the channel over this route to be leaked *)
    in(c,inter11:as);
    let as_root1 = get_root_as(rootwrap) in
      event Routing(AS, inter11, as_root1);
      event C_as(inter11);
      event C_routing(ip_res,get_root_ip(rootwrap));
      out(c,resolverchan2))
    |
    ((* get request for name resolution from client *)
    in(c, req:bitstring);
    let real_pack = get_req_packet(resolverchan4, req) in
        (* send request for authoritave name server to root NS *)
        out(c, req_packet(resolverchan,real_pack));
        in(c,rootans:bitstring);
        let root_pack = get_req_packet(resolverchan2, rootans) in
            let dnsinfo = get_ask(root_pack) in
              event UsedDomainServer(ip_res,get_dns_ip(dnsinfo));
              let dns_ip = get_dns_ip(dnsinfo) in
              (* Build private channel over IPs, the fresh src port and the *)
              (* public target port with the NS server information by the root NS. *)
              let resdnschan = chanbuilder(ip_res,port_res,dns_ip,NS_PORT) in
              in(c, dns_port_packet:com);
              let dns_port = d_get_src_port(dnsinfo,dns_port_packet) in
              let resdnschan2 = chanbuilder(dns_ip,dns_port,ip_res,RES_PORT) in
              (* Send request for MX record *)
              out(c, req_packet(resdnschan,real_pack));
              in(c, ans:bitstring);
              let real_pack3 = get_ans_packet(resdnschan2, ans) in
                let prov_ans = get_prov_answer(real_pack3) in
                if prov_ans = get_request_prov(real_pack) then
                let dom_ans = get_dom_answer(real_pack3) in
                event queries(client_ip,dom_ans);
                (* Send request for A record of the domain priory received *)
                out(c, req_packet(resdnschan,request_dom(dom_ans)));
                in(c, ans2:bitstring);
                let real_pack4 = get_ans_packet(resdnschan2, ans2) in
                if dom_ans = get_answer_dom(real_pack4) then
                  let answer_final = answer(real_pack,get_answer_ip(real_pack4),get_answer_dom(real_pack4)) in
                  event nDNSSEC(prov_ans);
                    out(c, ans_packet(resolverchan3,answer_final))
                    
        
    )
    |
    ((* get request for name resolution from client *)
    in(c, req_dnssec:bitstring);
    let real_pack = get_req_packet(resolverchan4, req_dnssec) in
    (* send request for authoritave name server to root NS using the DNSSEC protocol *)
        out(c, req_packet(resolverchan,real_pack));
        in(c,rootans:bitstring);
        (*receive DS record from root server *)
        let root_pack = get_req_packet(resolverchan2, rootans) in
            let dnssec_resp(ds_ans,pair(root_zsk,root_ksk),rr_sig_r,ds,ds_sig) = root_pack in
                (* get root key from the trustbase and continue to verify the message *)
                  get trustbase(=root_port, true_ksk_pub) in
                    if true_ksk_pub = root_ksk then
                      if verify(root_ksk, pair(root_zsk,root_ksk), rr_sig_r)=true then
                        let pk_dnssec = ds_snd(ds_ans) in
                          if ds_ans = ds then
                            if verify(root_zsk, ds, ds_sig) = true then
                              let dnswrapper = ds_fst(ds_ans) in
              let dns_ip = get_dns_ip(dnswrapper) in
              (* Build private channel over IPs, the fresh src port and the *)
              (* public target port with the NS server information by the root NS. *)
              let resdnschan = chanbuilder(ip_res,port_res,dns_ip,NS_PORT) in
              in(c, dns_port_packet:com);
              let dns_port = d_get_src_port(dnswrapper,dns_port_packet) in
              let resdnschan2 = chanbuilder(dns_ip,dns_port,ip_res,RES_PORT) in
              event UsedDomainServer(ip_res,dns_ip);
              let provreq = get_request_prov(real_pack) in
              (* Send request for MX record using the DNSSEC protocol *)
              out(c, req_packet(resdnschan,dnssec_request_prov(provreq)));
              in(c, ans:bitstring);
              let comb_message = get_ans_packet(resdnschan2, ans) in
                (* Start to verify the MX record using the key received from the root NS*)
                  let dnssec_resp(MX,pair(top_zsk1,top_ksk1),mx_sig_t,mx_rec,mx_sig) = comb_message in
                  let prov_resp = mx_fst(MX) in
                    if prov_resp = provreq then
                        if pk_dnssec = top_ksk1 then
                          if verify(top_ksk1, pair(top_zsk1,top_ksk1), mx_sig_t)=true then
                              if mx_rec = MX then
                                if verify(top_zsk1, mx_rec, mx_sig) = true then
                                  let dom_ans = mx_snd(MX) in
              event queries(client_ip,dom_ans);
              (* Send request for A record  using the DNSSEC protocol *)
              out(c, req_packet(resdnschan,dnssec_request(dom_ans)));
              in(c, ans2:bitstring);
              let comb_message2 = get_ans_packet(resdnschan2, ans2) in
                (* Start to verify the A record using the key received from the root NS*)
                  let dnssec_resp(A,pair(top_zsk2,top_ksk2),a_sig_t,a_rec,a_sig) = comb_message2 in
                  let dom_resp = rr_fst(A) in
                    if dom_resp = dom_ans then
                        if pk_dnssec = top_ksk2 then
                          if verify(top_ksk2, pair(top_zsk2,top_ksk2), a_sig_t)=true then
                              if a_rec = A then
                                if verify(top_zsk2, a_rec, a_sig) = true then
                                  let ip_ans = rr_snd(A) in
                                  (*Finally send the answer of the name resolution *)
                                  (* to the client process *)
                                  let real_pack_dnssec = answer(request(provreq),ip_ans,dom_ans) in
                                    out(c, ans_packet(resolverchan3,real_pack_dnssec))
              )
    ))
    .


let smtp_client(prov:provider,dom_client:dom,ip_client:ip,AS:as) =
    (* Receive public information of the Recipient MX process *)
    in(c, recipient_info:provider);
    (* Receive public information of the resolver process *)
    in(c, resolver_info:service);
    !(
    (* Choose fresh source port for each session *)
    new port_client:port;
    out(c,port_client);
    (* Send src port over trusted channel to communication partner. The previous line seems redundant since the 
    adversary can gain the same knowledge as the key to the private constructor is publically known. We keep this redundance because
    the first line expresses that the port is publically known, whereas the second represents that it is sent over a trusted channel. *)
    out(c,c_communicate_src_port(client_info(prov,AS,ip_client,dom_client),port_client));
    let recipient_prov = get_prov_valid(recipient_info) in
    let resolver_ip = get_res_ip(resolver_info) in
    in(c, res_packed_src_port:com);
    let res_port = r_get_src_port(resolver_info,res_packed_src_port) in
    (* Build private channel over IPs, the fresh src port and the *)
    (* public target port. *)
    let clientchan = chanbuilder(ip_client, port_client, resolver_ip, RES_PORT) in
    let clientchan2 = chanbuilder(resolver_ip, res_port,ip_client, CLIENT_PORT) in
    (
    ((* Corrupted IP -> Keys and Channels are leaked to the adversary *)
    event C_ip(ip_client);
    out(c,clientchan);
    out(c,clientchan2))
    |
    ((* Corrupted AS -> all IPs located in this AS are leaked -> *)
    (* Corrupted IP -> Keys and Channels are leaked to the adversary *)
    event C_as(AS);
    event C_ip(ip_client);
    out(c,clientchan);
    out(c,clientchan2))
    |
    ((* Adversary chooses an AS to be corrupted and chooses to route *)
    (* over this AS. This leaks the channel over this route to be leaked *)
    in(c,inter1:as);
    let resolver_as1 = get_res_as(resolver_info) in
      event Routing(AS, inter1, resolver_as1);
      event C_as(inter1);
      event C_routing(ip_client,resolver_ip);
      out(c,clientchan))
    |
    ((* Adversary chooses an AS to be corrupted and chooses to route *)
    (* over this AS. This leaks the channel over this route to be leaked *)
    in(c,inter11:as);
    let resolver_as2 = get_res_as(resolver_info) in
      event Routing(resolver_as2, inter11, AS);
      event C_as(inter11);
      event C_routing(ip_client,resolver_ip);
      out(c,clientchan2))
    |
    (event queries_prov(ip_client,recipient_prov);
    event Resolver(ip_client,resolver_ip);
    (* Goal: Send mail to recipient. *)
    (* Send request for domain resolution to resolver *)
    out(c,req_packet(clientchan,request(recipient_prov)));
    (* Receive mailserver domain and IP from resolver  *)
    in(c, answeri:bitstring);
      let real_pack = get_ans_packet(clientchan2, answeri) in
        let req = check_answer_req(real_pack) in
          if req = request(recipient_prov) then
            let IP = get_answer_ip(real_pack) in
              let recipient_dom = get_answer_dom(real_pack) in
                event Received(recipient_prov,recipient_dom,IP);
                in(c, server_packed_src_port:com);
                let server_port = s_get_src_port(IP,server_packed_src_port) in
                (* Build private channel over the received IP, the fresh src port *)
                (* and the public target port. *)
                let clientchan3 = chanbuilder(ip_client,port_client,IP,SERVER_PORT) in
                let clientchan4 = chanbuilder(IP, server_port, ip_client, CLIENT_PORT) in
                (
                ((* Corrupted IP -> Keys and Channels are leaked to the adversary *)
                event C_ip(ip_client);
                out(c,clientchan3);
                out(c,clientchan4))
                |
                ((* Adversary chooses an AS to be corrupted and chooses to route *)
                (* over this AS. This leaks the channel over this route to be leaked *)
                event C_as(AS);
                event C_ip(ip_client);
                out(c,clientchan3);
                out(c, clientchan4))
                |
                (
                (* freshly choose challenge mail to send *)
                new mail:bitstring;
                out(c,req_packet(clientchan3,smtp_req(ip_client,IP,mail)));
                in(c, answer_pak:bitstring);
                (* Receive acknowledgement of the receiver *)
                if OK = get_ans_packet(clientchan4,answer_pak) then
                (* Wait for adversary to send the challenge mail *)
                (* If the adversary can provide the challenge, *)
                (* the mail delivery is unconfidential *)
                in(c, challenge:bitstring);
                if challenge=mail then
                    event Unconf(prov, recipient_prov))
                )
    )))
    .


let smtp_server(prov:provider,dom_server:dom,IP:ip,AS:as) =
    (* Receive public information of the Client MX process *)
    in(c, clientinfo:service);
    !(
    (* Choose fresh source port for each session *)
    new server_port:port;
    out(c, server_port);
    (* Send src port over trusted channel to communication partner. The previous line seems redundant since the 
    adversary can gain the same knowledge as the key to the private constructor is publically known. We keep this redundance because
    the first line expresses that the port is publically known, whereas the second represents that it is sent over a trusted channel. *)
    out(c, s_communicate_src_port(IP,server_port));  
    let client_ip = get_client_ip(clientinfo) in
    in(c, client_port:port);
    (* Build private channel over IPs, the fresh src port and the *)
    (* public target port. *)
    let serverchan = chanbuilder(client_ip, client_port, IP, SERVER_PORT) in
    let serverchan2 = chanbuilder(IP, server_port, client_ip, CLIENT_PORT) in
    (
     ((* Corrupted IP -> Keys and Channels are leaked to the adversary *)
      event C_ip(IP);
      out(c,serverchan);
      out(c,serverchan2))
     |
     ((* Corrupted AS -> all IPs located in this AS are leaked -> *)
      (* Corrupted IP -> Keys and Channels are leaked to the adversary *)
      event C_as(AS);
      event C_ip(IP);
      out(c,serverchan);
      out(c,serverchan2))
     |
     ((* Adversary chooses an AS to be corrupted and chooses to route *)
      (* over this AS. This leaks the channel over this route to be leaked *)
      in(c,inter1:as);
      let as_c1 = get_client_as(clientinfo) in
      event Routing(as_c1, inter1, AS);
      event C_as(inter1);
      event C_routing(get_client_ip(clientinfo),IP);
      out(c,serverchan))
    |
    ((* Adversary chooses an AS to be corrupted and chooses to route *)
      (* over this AS. This leaks the channel over this route to be leaked *)
      in(c,inter2:as);
      let as_c2 = get_client_as(clientinfo) in
      event Routing(as_c2, inter2, AS);
      event C_as(inter2);
      event C_routing(get_client_ip(clientinfo),IP);
      out(c,serverchan2))
    |
    ( (* Receive Message from the client MX *)
    in(c, message:bitstring);
    let real_pack = get_req_packet(serverchan, message) in
            let mail = smtp_third(real_pack) in
            (* acknowledge receiption of the message *)
            out(c, ans_packet(serverchan2,OK)))
    ))
    .

process
(
  !(in(c,prov:provider);
    !(
      in(c, dom_c:dom);
      in(c, ip_c:ip);
      in(c, AS_c:as);
      event isMailserver(dom_c,prov);
      event IPinAS(ip_c,AS_c);
      event A_record(ip_c,dom_c);
      out(c,client_info(prov,AS_c,ip_c,dom_c));
      (* *)
      !(smtp_client(prov,dom_c,ip_c,AS_c))
      )
    |
    !(
      in(c, dom_s:dom);
      in(c, ip_s:ip);
      in(c, AS_s:as);
      event isMailserver(dom_s,prov);
      event IPinAS(ip_s,AS_s);
      event A_record(ip_s,dom_s);
      event Register_A(dom_s,ip_s);
      event Register_MX(prov,dom_s);
      out(c,server_info(prov,AS_s,ip_s,dom_s));
      out(c,register_server(dom_s,ip_s));
      out(c,register_provider(prov,dom_s));
      out(c,valid_prov(prov));
      (* *)
      !(smtp_server(prov,dom_s,ip_s,AS_s))
      )
    |
    !(
      in(c, dom_r:dom);
      in(c, ip_r:ip);
      in(c, AS_r:as);
      event IPinAS(ip_r,AS_r);
      event A_record(ip_r,dom_r);
      out(c,res_info(dom_r,ip_r,AS_r));
      (* *)
      !(resolver(dom_r,ip_r,AS_r))
      )
    |
    !(
      in(c, dom_d:dom);
      in(c, ip_d:ip);
      in(c, AS_d:as);
      event IPinAS(ip_d,AS_d);
      event A_record(ip_d,dom_d);
      out(c,dns_info(dom_d,ip_d,AS_d));
      (* *)
      !(dns_server(dom_d,ip_d,AS_d))
      )
    |
    !(
      in(c, dom_root:dom);
      in(c, ip_root:ip);
      in(c, AS_root:as);
      event IPinAS(ip_root,AS_root);
      event A_record(ip_root,dom_root);
      out(c,root_info(dom_root,ip_root,AS_root));
      (* *)
      !(root_server(dom_root,ip_root,AS_root))
      )
    )
  )
\end{lstlisting}
\end{full}

\end{document}